\documentclass[12pt]{article}

\usepackage[T1]{fontenc}

\usepackage{amsmath, amssymb, amsthm}

\usepackage{graphicx}
\usepackage{xcolor}
\usepackage[utf8]{inputenc}

\usepackage{booktabs}
\usepackage{multirow}
\usepackage{siunitx}
\usepackage{float}
\usepackage{threeparttable}
\usepackage{caption,subcaption}

\usepackage{enumerate}
\usepackage{placeins}
\usepackage{url}

\usepackage[authoryear]{natbib}

\usepackage[colorlinks=true,citecolor=blue,linkcolor=blue]{hyperref}

\newcommand{\blind}{1}

\newcommand{\argmin}{\mathop{\mathrm{argmin}}}

\makeatletter
\theoremstyle{plain}
\newtheorem{assumption}{\protect\assumptionname}
\newtheorem{lem}{\protect\lemmaname}
\newtheorem{thm}{\protect\theoremname}

\theoremstyle{remark}
\newtheorem{rem}{\protect\remarkname}
\newtheorem{definition}{Definition}

\providecommand{\assumptionname}{Assumption}
\providecommand{\corollaryname}{Corollary}
\providecommand{\lemmaname}{Lemma}

\providecommand{\remarkname}{Remark}
\providecommand{\theoremname}{Theorem}
\providecommand{\examplename}{Example}

\definecolor{dgreen}{rgb}{0, 0.7, 0.0}
\makeatother

\begin{document}

\date{January 2026}

\def\spacingset#1{\renewcommand{\baselinestretch}
{#1}\small\normalsize} \spacingset{1}

\if1\blind
{
  \title{\bf Forecasted Treatment Effects \\ with Short Panels\thanks{We thank Isaiah Andrews, Xavier D'Haultfoeuille, Chris Muris, Krishna Pendakur, and seminar and conference participants at several venues for helpful comments. We are deeply grateful to Laura Zell (University of Lucerne) for developing the R package fatEstimator, which implements the FAT estimator used in this paper and can be found at \url{https://github.com/LauraZell/FAT_package}. Irene Botosaru gratefully acknowledges financial support from the Social Sciences and Humanities Research Council of Canada IG 435-2021-0778 and the Canada Research Chairs Program. Martin Weidner gratefully acknowledges financial support by  the
European Research Council grant ERC-2018-CoG-819086-PANEDA. 
}}
  \author{\setcounter{footnote}{2}
  Irene Botosaru\thanks{McMaster University, Department of Economics. Email: {\tt botosari@mcmaster.ca}}, Raffaella Giacomini\thanks{University College London, Department of Economics/Cemmap. Email: {\tt r.giacomini@ucl.ac.uk}}, Martin Weidner\thanks{Oxford University, Department of Economics and Nuffield College. Email: {\tt martin.weidner@economics.ox.ac.uk}}}
  \maketitle
} \fi

\if0\blind
{
  \bigskip
  \bigskip
  \bigskip
  \begin{center}
    {\LARGE\bf Forecasted Treatment Effects with Short Panels}
\end{center}
  \medskip
} \fi

\bigskip
\begin{abstract}
\noindent
We consider estimation and inference of the effects of a policy in the absence of an untreated or control group. We obtain unbiased estimators of individual (heterogeneous) treatment effects and a consistent and asymptotically normal estimator of the average treatment effect. Our estimator averages, across individuals, the difference between observed post‑treatment outcomes and unbiased forecasts of their counterfactuals, based on a (short) time series of pre-treatment data. The paper emphasizes the importance of focusing on forecast unbiasedness rather than accuracy when the end goal is estimation of average treatment effects. We show that simple basis function regressions ensure forecast unbiasedness for a broad class of data generating processes for the counterfactuals. In contrast, forecasting based on a specific parametric model requires stronger assumptions and is prone to misspecification and estimation bias.  We show that our method can replicate the findings of some previous empirical studies but it does so without using an untreated or control group.

\end{abstract}

\vfill

\newpage
\spacingset{1.6}

\section{Introduction}

Evaluating the effect of a policy or a treatment usually requires comparing treated and untreated units. Standard approaches such as difference-in-differences (DiD) or two-way fixed effects (TWFE) regressions rely on outcomes of untreated units to identify the counterfactual outcomes of treated units. These methods rely on assumptions such as unconfoundedness or parallel trends (e.g., \citealt{GoodmanBacon, CallawaySantAnna, SunAbraham}). They cannot be applied when all units receive the treatment.

This situation is common. Examples include national tax reforms, health programs, and environmental regulations that affect all individuals or regions at the same time. Even when adoption is staggered, the last periods after full treatment do not contain untreated units, so conventional DiD or event-study designs cannot recover treatment effects for those periods. In such cases, researchers often turn to structural models or to simple before–after comparisons. These settings call for methods that make the required assumptions explicit and allow researchers to evaluate their plausibility.

This paper studies how to estimate treatment effects when no valid control group exists. We propose a transparent and practical approach based on forecasting each treated unit’s counterfactual outcome from its own pre-treatment data using basis function regressions. The key idea is that in short panels, the property required for identification is forecast unbiasedness rather than forecast accuracy. Forecasting methods in time-series and machine learning aim to minimize prediction error and often allow some bias to reduce variance. In policy evaluation, this objective is not appropriate. An unbiased forecast of the counterfactual outcome guarantees an unbiased estimate of the treatment effect, while a forecast that is accurate but biased does not. Of course, our approach relies on its own identifying assumptions, which we make explicit, but they differ from those used in designs that compare treated and untreated units.

Our approach, the Forecasted Average Treatment effect (FAT) estimator, computes the average treatment effect on the treated as the average difference between the observed post-treatment outcome and an unbiased forecast of the outcome in the absence of treatment. What matters is that the forecasting rule produces unbiased predictions on average across individuals.

The paper makes three contributions. First, we develop a simple and general framework for estimating treatment effects when no untreated units are available. Under appropriate assumptions, the FAT estimator is consistent and asymptotically normal even when the time dimension is short, the panel is unbalanced, and treatment effects are heterogeneous.

Second, we characterize a broad class of data generating processes (DGPs) under which the required forecast unbiasedness condition holds. When individual counterfactual outcomes can be generally expressed as the sum of (at most) three unobserved components - a stationary process, a unit root process, and a deterministic trend - unbiased forecasts can be obtained by regressing pre-treatment outcomes on basis functions of time, such as low-order polynomials. This result allows for DGPs with fixed effects, heterogeneous autoregressive parameters, and non-stationary trends, without requiring a detailed specification of the stochastic component. The framework also accommodates heterogeneous dynamic processes, such as unit-specific autoregressive parameters, which are difficult to handle with standard short-panel methods. We call our proposed approach ``Unobserved-Components FAT''.

Third, we compare our general ``Unobserved-Components FAT'' approach with what we refer to as ``Parametric-Model FAT'', where the researcher specifies and estimates a particular parametric model such as an AR(1) to forecast counterfactuals. This approach relies on stronger assumptions and can perform poorly in short panels because of misspecification or estimation bias. In contrast, our regression-based approach that focuses directly on forecast unbiasedness is robust to specification errors and is not affected by the incidental parameter problem.

The FAT estimator can also be used when untreated units exist but are imperfect comparators. In such cases, FAT can provide a robustness check for studies that rely on, e.g., parallel-trends assumptions. The framework complements recent work on DiD and event-study designs (e.g., \citealt{deChaisemartinX2020, CallawaySantAnna, SunAbraham}) by providing an alternative benchmark based on explicit assumptions about the individual-specific counterfactual process rather than on cross-group comparability.

Crucially, the approach based on FAT is advantageous when outcomes exhibit dynamics. 
Recent work \citep{MarxTamerTang, Klosin2024, BotosaruLiu2025, Cornwall2025} 
has shown that traditional TWFE and dynamic panel estimators often yield biased or 
negatively weighted estimates of the treatment effect in the presence of serial 
correlation. The solutions proposed in this literature typically rely on the 
\textit{homogeneity} of the persistence parameters across units. By construction, 
FAT avoids conflating structural treatment effects with outcome persistence and is 
robust to feedback mechanisms \citep{Bonhomme2025}, while allowing for a much 
richer class of models, which include \textit{heterogeneous} forms of serial 
correlation, where the degree of outcome persistence can vary arbitrarily across 
individuals.

This paper relates to several strands of literature. It is thematically related to work on causal inference in settings where contemporaneous untreated comparison groups are not the primary source of identification, and where time-series structure plays a central role in constructing counterfactual outcomes from pre-treatment data. These include Bayesian approaches to causal forecasting (e.g., \citealt{Brodersen2015}); (comparative) interrupted time-series methods, in which untreated potential outcomes are modeled using finite-dimensional parametric mean functions in time—typically linear or piecewise-linear trends with level and/or slope changes at the intervention date (e.g., \citealt{Schochet2022,Bernaletal2017,BrownWarner1985}); and methods based on data-driven aggregation to purge unobserved aggregate confounders in designs with aggregate shocks (e.g., \citealt{ArkhangelskyKorovkin2023}).

The FAT framework differs from these strands along three dimensions. First, relative to Bayesian causal forecasting, FAT is designed to deliver forecast-unbiased counterfactuals under high-level restrictions on the untreated outcome process and does not require the researcher to commit to a fully specified likelihood or prior. Second, relative to interrupted time-series methods, FAT strictly nests the canonical segmented-regression mean specifications while allowing for richer latent evolution, including stochastic trend components and heterogeneous unit-specific dynamics. Third, relative to the aggregate-shock/aggregation literature, FAT does not take identification from exogenous aggregate shocks or instruments, nor does it rely on cross-sectional weighting to eliminate unobserved aggregate confounding; instead, it is a short-panel forecasting approach in which identification is driven by the internal time-series information in treated units’ pre-intervention histories.

In addition, the paper contributes to the panel-data literature on heterogeneous dynamic processes and random coefficients (e.g., \citealt{Chamberlain1992,ArellanoBonhomme2012,GrahamPowell2012}). For example, the projection estimator in \cite{ArellanoBonhomme2012} can be interpreted as a special case of FAT under specific restrictions on the latent outcome process. In particular, for a unit-root process without drift, the Arellano–Bonhomme forecast corresponds to an equal-weighted average of pre-treatment outcomes, whereas the FAT polynomial forecast assigns greater weight to more recent observations. When deterministic trends are present, the projection approach in \cite{ArellanoBonhomme2012} requires correct specification of the trend component, while FAT only requires the choice of a sufficiently rich family of basis functions, such as low-order polynomials. Therefore, FAT delivers unbiased counterfactual forecasts under weaker restrictions on the trend component, while retaining the projection-based intuition underlying existing random-coefficient panel estimators. 

Our analysis also relates to work on forecast unbiasedness in time-series models (e.g., \citealt{FullerHasza, Dufour1984}) and to the dynamic-panel literature on bias correction and consistent estimation with short time dimensions (e.g., \citealt{AngristPischke,BlundellBond1998}). Our approach differs by showing that consistent treatment effect estimation in short panels can be achieved through unbiased forecasts rather than through the specification of a full stochastic model.

The rest of the paper is organized as follows. Section~\ref{sec:Baseline-model} presents the framework and the FAT estimator. Section~\ref{sec:Covariates} extends the method to include covariates and discusses forecasting based on a parametric model. Section~\ref{sec:mc} reports simulation results. Section~\ref{sec:EmpiricalIllustration} provides an empirical illustration. Section~\ref{sec:Conclusion} concludes.

\section{Baseline case: Unobserved-Components FAT\label{sec:Baseline-model}}

In this section we develop our baseline “Unobserved-Components FAT” estimator for the average treatment effect on the treated (ATT) under universal treatment. We use individual pre‑treatment time series to forecast each unit’s post‑treatment counterfactual outcome and then average, across units, the difference between the observed outcome and its forecast. We first establish consistency and asymptotic normality of this estimator under a high‑level forecast‑unbiasedness condition, and then show that this condition holds for a broad class of unobserved‑components DGPs when counterfactuals are forecast using basis‑function regressions.

\subsection{Parameter of interest and estimator}

Consider a treatment or a policy that is implemented at a time $\tau$.\footnote{In Section \ref{stagger}, we discuss individual-specific timing.} Here, the treatment affects all individuals in the population at the same time, so that the treatment indicator of individual $i$ at time $t$ is given by $$d_{it}:=\ 1(t>\tau) \text{ for all } i=1,\dots, n.$$
We adopt the potential outcomes framework with each individual $i$ having two potential outcomes at each time $t$: $y_{it}(1)$ if the individual is exposed to the treatment and $y_{it}(0)$ if the individual is not exposed to the treatment. Due to the absence of an untreated group in our setting, we will henceforth simply refer to $y_{it}(0)$ as the ``counterfactual". 

Under the stable unit treatment value assumption (SUTVA), the observed outcome of individual $i$ at $t$ is: 
\begin{align*}
y_{it} =(1-d_{it})\,y_{it}(0)+d_{it}\,y_{it}(1).
\end{align*}
When all individuals are treated after time $\tau$, we have 
 \begin{align}
    y_{it} & =\left\{
    \begin{array}{ll} y_{it}(0) & \text{for}\;t\leq\tau,\\
    y_{it}(1) & \text{for}\;t>\tau.
    \end{array}\right.
    \label{obsoutcome}
\end{align}
We follow the literature on heterogeneous treatment effects in defining the ATT $h\geq 1$ periods after $\tau$ as: 
\begin{align}
{\rm ATT}_{h} & :=\frac{1}{n}\sum_i \mathbb{E}\left[y_{i\tau+h}(1)-y_{i\tau+h}\left(0\right)\right]\label{ATT}\\
 & =\frac{1}{n}\sum_i\mathbb{E}\left[y_{i\tau+h}-y_{i\tau+h}\left(0\right)\right],
\end{align}
where we used that $y_{i\tau+h}(1)=y_{i\tau+h}$ for $h\geq1$.\footnote{Note that with identically and independently distributed data across $i$, the right hand side of \eqref{ATT} reduces to the conventional $\mathbb{E}\left[y_{i\tau+h}-y_{i\tau+h}\left(0\right)\right]$.}

The challenge in identifying and estimating ${\rm ATT}_h$
is that the counterfactual $y_{i\tau+h}\left(0\right)$ is not observed
for $h \geq 1$. The conventional
approach in the presence of an untreated group is to impose sufficient assumptions that identify the parameter of interest from the observed post-treatment outcomes of the untreated group. In the absence of
an untreated group, we exploit pre-treatment individual time series to obtain
a forecast for $y_{i\tau+h}\left(0\right)$. We denote this forecast by $\widehat{y}_{i\tau+h}(0)$.

We call our proposed estimator for ${\rm ATT}_{h}$ the Forecasted Average Treatment effect estimator (FAT), defined as:
\begin{align}
\widehat{\rm FAT}_{h}:=\frac{1}{n}\sum_{i=1}^{n}\left[y_{i\tau+h}-\widehat{y}_{i\tau+h}\left(0\right)\right],\label{DefFAT}
\end{align}
where $\widehat{y}_{i\tau+h}(0)$ is a measurable function of past outcomes $\left\{y_{it}\right\}_{t\leq \tau}$.\footnote{In the baseline case, the individual forecast depends only
on the past outcomes of the treated, in particular, there are no covariates in the information set.} We explain how to obtain the forecast $\widehat{y}_{i\tau+h}(0)$ below. For now, we note that $\widehat{y}_{i\tau+h}(0)$ uses individual-specific pre-treatment outcomes, which naturally accommodates unbalanced panels and heterogeneous treatment effects.

We make the following high-level assumptions. 

\begin{assumption}[Average unbiasedness]
    \label{Unbiasedness}
    The forecast for time $\tau+h, \; h\geq1,$ is unbiased on average, in the sense that:
    \begin{align}
    \frac{1}{n}\sum_i\mathbb{E}\left(\widehat{y}_{i\tau+h}\left(0\right)-y_{i\tau+h}\left(0\right)\right)=0.
    \end{align}
\end{assumption}

Note that Assumption~\ref{Unbiasedness} guarantees that
$\mathbb{E}(\widehat{{\rm FAT}}_{h}) = {\rm ATT}_{h}$ so that ${\rm ATT}_{h}$ is identified.\footnote{The result follows trivially by writing ${\rm ATT}_{h} =\frac{1}{n}\sum_i\mathbb{E}\left[y_{i\tau+h}-\widehat{y}_{i\tau+h}(0) + \widehat{y}_{i\tau+h}(0) - y_{i\tau+h}\left(0\right)\right].$}

Let $u_{i\tau+h}:=y_{i\tau+h}-\widehat{y}_{i\tau+h}(0)$ be the forecasted individual treatment effect at $\tau+h, \; h\geq 1$.

\begin{assumption}[CLT]
    \label{ass:Sampling} 
Let $\{u_{i\tau+h}\}$ be a sequence of random variables that satisfies a CLT:
\begin{equation}
        \frac{\frac{1}{\sqrt{n}}\sum_i 
        \left( u_{i\tau+h}
       { - \mathbb{E}u_{i\tau+h}}
        \right)
        }{\bar{\sigma}_n} \Rightarrow{\cal N}\left(0,1\right),
    \end{equation}
    where $\bar{\sigma}^2_n:= {\rm Var}(\frac{1}{\sqrt{n}} \sum_i u_{i\tau+h})<\infty$. \label{V}
\end{assumption}
For example, when $\{u_{i\tau+h}\}$ is a sequence of cross-sectionally independent (but not identically distributed) random variables, Theorem 5.11 in \cite{White} gives an asymptotic normality result. 

What Assumptions 1 and 2 fundamentally rule out is shocks that affect all individuals after the treatment and that are unforecastable.\footnote{In principle, the assumptions allow for some common shocks to be captured by the method used to forecast the counterfactuals. For example, if the common shocks are deterministic and can be modeled as polynomials, our baseline method will be able to account for it.} This is an example of the assumptions that any method for treatment effect evaluation must inevitably make in the absence of a valid control group. The assumptions are more likely to be valid the higher the frequency at which the data is measured and the closer to the treatment date one evaluates the effect (i.e., they more plausibly hold for monthly data and for effects evaluated one month after the treatment than for yearly data and effects evaluated several years after the treatment). Below we provide some discussion of how the assumption of no common shocks after the treatment could be weakened. As one would expect, the assumption can be relaxed when there is a group of untreated individuals that are subject to the same unforecastable shock, as we discuss in the Online Appendix (Section \ref{controlgroup}). Perhaps less obviously, the assumption could also be relaxed in the absence of an untreated group, as long as we observe another variable for the treated units that is subject to the same unforecastable shock but not subject to the treatment (see the discussion in Section \ref{NCcommonshock}).

\begin{lem}[Consistency and asymptotic normality]
    \label{lem:LemmaConsistencyFAT} 
    For each $i=1,\ldots,n$, let the forecast $\widehat{y}_{i\tau+h}(0)$, $h\geq 1$, be a function of $\{y_{it}\}_{t\leq\tau}$. Let Assumptions \ref{Unbiasedness} and \ref{ass:Sampling} hold. Then\footnote{
Estimation of the variance $\bar{\sigma}_n^2$
is discussed in Appendix~\ref{app:VarEstimation}.
    }
    \begin{align*}
    \frac{\sqrt{n}\left(\widehat{{\rm FAT}}_{h}-{\rm ATT}_{h}\right)}{\bar{\sigma}_n}\Rightarrow{\cal N}\left(0,1\right).
    \end{align*}
\end{lem}

In the remainder of the paper, we provide low-level sufficient assumptions, including a full description of the class of DGPs for the counterfactuals $y_{it}(0)$ and the forecast methods which satisfy Assumption \ref{Unbiasedness}.

\subsection{Unbiased forecasts of counterfactuals} \label{WeightedAverage}

In this section, we characterize the class of DGPs for the counterfactuals. The need to discuss the DGP for counterfactuals arises because of the lack of an untreated group, which means that we must rely on forecasting counterfactuals from pre-treatment observations. 

\subsubsection{Stationary or stochastic trends DGPs}

In this section, we consider a class of DGPs such that Assumption \ref{Unbiasedness} is satisfied generally, namely, by any forecast that can be written as a weighted average of pre-treatment outcomes with weights summing to 1. The DGPs in this class express the counterfactual as the sum of potentially two unobserved stochastic components. This includes a variety of processes, such as stationary and non-stationary (unit root) ARMA processes with individual-specific parameters.

\begin{assumption}[Stationary or stochastic trends DGPs]
\label{ass:Series} 
    The counterfactual $y_{it}(0)$ is:
    \begin{align}
    y_{it}(0) & = y_{it}^{(1)}(0)+ y_{it}^{(2)}(0),\label{sumprocess}
    \end{align}
    where $y_{it}^{(1)}(0)$ is an unobserved mean-stationary process (i.e., $\mathbb{E}y_{it}^{(1)}(0)$ is constant over t)  and $y_{it}^{(2)}(0)=y_{it-1}^{(2)}(0)+u_{it}(0)$ is an unobserved random walk process with $\mathbb{E}u_{it}(0)=0$ for all $t\geq2$. Either/both components could be zero.
\end{assumption}

\begin{rem} Assumption \ref{ass:Series} does not require both components to be present, which means that it accommodates stationarity as well as non-stationarity due to a stochastic trend. The user does not need to take a stance on the component(s). Our method is robust to both. When both components in Assumption \ref{ass:Series} are present, the assumption is equivalent to the classical trend-cycle decomposition of macroeconomic time series with stochastic trends (e.g., \cite{NelsonPlosser1982, Watson1986}).  
\end{rem}
\begin{rem} This class of DGPs is a plausible assumption for applications where either (1) the time series of pre-treatment outcomes does not display a trend; (2) there is a trend in pre-treatment outcomes that is plausibly stochastic (not deterministic); or (3) there is only one pre-treatment observation so a deterministic trend could never be modeled anyway. We consider processes with deterministic trends in pre-treatment outcomes in the next section. 

\end{rem}

 A key insight of this paper is that a correctly specified parametric model (e.g., a specific ARIMA model) is not necessary to obtain unbiased forecasts of the counterfactuals. In fact, as the next result shows, any forecast expressed as a weighted average of pre-treatment observations satisfies the unbiasedness condition.

\begin{thm}[Unbiasedness for stationary or stochastic trends DGP]
    \label{th:Unbiasedness1} 
    Let Assumption \ref{ass:Series} hold. Denote by ${\cal T}_{i}=\left\{ \tau-R_{i}+1,\ldots,\tau\right\}$ the set of $R_{i}$ time periods directly preceding the treatment date. Consider a weighted average of the pre-treatment outcomes:
\begin{align}
    \widehat{y}_{i\tau+h}(0)=\sum_{t\in{\cal T}_{i}}w_{it}y_{it},\label{weightavg}
\end{align} 
where $w_{it}$ are non-random weights such that $\sum_{t\in{\cal T}_{i}}w_{it}=1$. 
    Then, 
    \begin{align}\label{strongunb}
    \mathbb{E}\left[\widehat{y}_{i\tau+h}(0)-y_{i\tau+h}(0)\right] & =0.
    \end{align}
   
\end{thm}

\begin{rem}
Note that Theorem \ref{th:Unbiasedness1} shows how to obtain unbiased estimates of the individual (possibly heterogeneous) treatment effects. The result in \eqref{strongunb} is stronger than the average unbiasedness required by Assumption \ref{Unbiasedness}. This means that under the assumptions of Theorem \ref{th:Unbiasedness1} one can not only obtain consistent and asymptotically normal estimates of the average treatment effects, but also unbiased estimates of the \textit{individual} treatment effects.
\end{rem}
There are many ways to obtain forecasts that are weighted averages of pre-treatment data. In this paper we focus on a general class of forecasts obtained via basis function regressions, such as polynomial time trends regressions. 

\begin{definition}[Forecasts via basis function regressions] \label{basisdef}
Consider a sequence of linearly independent functions $\{b_k(t)\}_{k=0}^{q_i}$, $q_{i}\in\left\{0,1,2,\ldots,\tau-1\right\}$, on the set $\mathcal{T}_{i}=\left\{ \tau-R_{i}+1,\ldots,\tau\right\}$ with  $R_{i}\in\left\{ q_i+1,\ldots,\tau\right\}$, and such that $b_0(t)=1$ for all $t$. For example, polynomial time trends set $b_k(t)=t^k$, with $q_{i}$ the order of the polynomial. For each individual $i$, we forecast the counterfactual via individual-specific regressions of pre-treatment outcomes $\{y_{it}\}_{t\in\mathcal{T}_i}$ on the basis functions $\{b_k(t)\}_{k=0}^{q_i}$: 
\begin{align}    
    \widehat{y}_{i\tau+h}^{(q_i,R_i)} 
    & :=\sum_{k=0}^{q_i}\widehat{c}_{ik}^{(q_i,R_i)}b_k(\tau+h),\label{generalforecast}\\
    \widehat{c_i}^{(q_i,R_i)} 
    & :=\argmin_{c\in\mathbb{R}^{q_i+1}}\sum_{t\in{\cal T}_{i}}\left(y_{it}-\sum_{k=0}^{q_{i}}c_{k}\,b_{k}\left(t\right)\right)^{2}\label{OLScoefficients},
\end{align}
where $c_i=\left(c_{i0},\ldots,c_{iq_{i}}\right)$ is a
$q_{i}+1$ vector of individual-specific coefficients.\footnote{Note that when $q_i = \tau - 1$, $R_i = \tau$, that is, all pre-treatment outcomes are used in constructing the forecast. However, fewer observations can be used. We discuss the choice of the tuning parameters $q_i$ and $R_i$ in Section \ref{choicepoly}.}
\end{definition}

This definition makes it clear that for any type of basis function the choice $q_i=0$ yields the sample mean of pre-treatment outcomes as the forecast. The following result shows that forecasts obtained via basis function regressions satisfy the weighted average requirement of Theorem \ref{th:Unbiasedness1}.

\begin{lem}
\label{ForecastWeights}
For known basis functions $\{b_k(t)\}_{k=0}^{q_i},\,q_{i}=0,1,\dots,\tau_{i}-1$
that are linearly independent on $\mathcal{T}_{i}$ with $b_{0}\left(t\right)=1$,
the forecast in equation \eqref{generalforecast} satisfies  equation \eqref{weightavg}.
\end{lem}

Our framework connects to the literature on random coefficient panel models, particularly \cite{Chamberlain1992}, \cite{GrahamPowell2012}, and \cite{ArellanoBonhomme2012}. Linear panel models with unit-specific intercepts and trend coefficients, similar to specifications in this literature, belong to the class of DGPs in Assumption~\ref{ass:Series} or Assumption~\ref{ass:Series2}. However, our goals differ: that literature studies identification of the distribution of heterogeneous coefficients, whereas we focus on forecasting untreated potential outcomes to estimate average treatment effects.

\subsubsection{Deterministic trends DGPs}

In this section, we consider an expanded class of DGPs that is appropriate for applications where: 1) there is more than one pre-treatment outcome; 2) it makes sense to model the outcomes as trending over time; 3) the trend is deterministic rather than (or in addition to) stochastic. We show that the basis function regression considered in the previous section gives unbiased forecasts of the counterfactuals, under certain conditions.

The expanded class of DGPs always includes a deterministic trend component, possibly in addition to (either or both) the stochastic components considered in Assumption \ref{ass:Series}.

\begin{assumption}[Deterministic trend DGPs]
\label{ass:Series2} 
    The counterfactual $y_{it}(0)$ is:
    \begin{align}
    y_{it}(0) & = y_{it}^{(1)}(0)+ y_{it}^{(2)}(0)+y_{it}^{(3)}(0),\label{sumprocess2}
    \end{align}
    where $y_{it}^{(1)}(0)$ and $y_{it}^{(2)}(0)$ are as in Assumption \ref{ass:Series} and $y_{it}^{(3)}(0)$ is a deterministic time trend $y_{it}^{(3)}(0)=\sum_{k=0}^{q_{0i}}c_{ik}^{(3)}{b}_{k}(t)$ with $c_{i}^{(3)}\in\mathbb{R}^{q_{0i}+1}$ and known basis functions $\{b_k(t)\}_{k=0}^{q_{0i}}$, $q_{0i}\in\left\{0,1,2,\ldots,\tau-1\right\}$.
\end{assumption}

Theorem \ref{th:Unbiasedness2} below clarifies when forecasts obtained via basis function regressions satisfy the unbiasedness assumption (Assumption \ref{Unbiasedness}) in the presence of deterministic trends.

\begin{thm}[Unbiasedness for deterministic trend DGPs]
    \label{th:Unbiasedness2} 
    Let Assumption \ref{ass:Series2} hold. Then, 
    \begin{align*}    \mathbb{E}\left[\widehat{y}_{i\tau+h}^{(q_{i},R_{i})}(0)-y_{i\tau+h}(0)\right] & =0,
    \end{align*}
    where $\widehat{y}_{i\tau+h}^{(q_{i},R_{i})}(0)$ is defined in \eqref{generalforecast}. If $y_{it}^{(3)}(0)$ is not zero, it must be that $q_{i} \geq q_{0i}$.
\end{thm}

\begin{rem} While the stochastic components of the counterfactual in Assumption \ref{ass:Series2} are unobserved and do not have to be present, if the deterministic time trend component is present, it must be a function of the same basis functions used to obtain the forecast. This implies that the stochastic component of the DGP does not need to be correctly specified, but, if a deterministic time trend component is present, it is correctly specified up to the order (as we discuss in the next remark).  
This is an unusual result from the perspective of forecasting, where one typically focuses on specifying both stochastic and deterministic parts of a model.
\end{rem}

\begin{rem} The key requirement of Theorem \ref{th:Unbiasedness2} is that $q_i$ (the number of basis functions used in estimation) be greater than or equal to $q_{0i}$ (the true number of basis functions), when the DGP has a deterministic time trend component. Intuitively, this means that, in the presence of deterministic trends, choosing a too small number of basis functions runs the risk of delivering biased forecasts of the counterfactuals. We discuss the practical implications of these findings when discussing the choice of tuning parameters in Section \ref{choicepoly}.
 \end{rem}

\begin{rem} \label{NCcommonshock} It may be possible to weaken the assumption that there is no forecastable common shock between $\tau$ and $\tau+h, h\geq 1$. To see this, let $\widetilde{y}_{it}\left(0\right)$ follow either Assumption \ref{ass:Series} or \ref{ass:Series2} and let $W_{it}$ be an observed variable that is affected by a common shock $\gamma_{t}$ but not by the treatment. Suppose that 
\begin{align}
    y_{it}\left(0\right) &= \gamma_{t}1\left(t>\tau\right)+\widetilde{{y}}_{it}\left(0\right),\\
    W_{it}	&= \theta_{i}^{W}\gamma_{t}1\left(t>\tau\right)+\epsilon_{it}.
\end{align}
Denote by $\widehat{\gamma}_t$ an estimator of $\gamma_t$ obtained from post-treatment data on $W_{it},$ $t>\tau$ via, e.g., PCA. Then 
\begin{align}
    \widehat{{\rm FAT}}_{h}&= \frac{1}{n}\sum_i \left(y_{i\tau+h}-\widehat{\widetilde{y}}_{i\tau+h}\left(0\right)\right)-\widehat{\gamma}_{\tau+h},
\end{align}
where $\widehat{\widetilde{y}}_{i\tau+h}(0)$ is the counterfactual forecast obtained as before via basis function regression.

Leveraging auxiliary information to partial out the effect of unobserved aggregate shocks is conceptually related to  \cite{WhiteKennedy2009, Freyaldenhoven2019, BrownBW2023}, among others.
\end{rem}
 
\subsection{Choice of basis functions and tuning parameters}\label{choicepoly}
Our proposed method for forecasting counterfactuals in Definition \ref{basisdef} requires choosing: 1) the number of pre-treatment periods used for the estimation $R_i$; 2) the type of basis functions $b_k(t)$; and 3) the number of basis functions $q_i$. We discuss how these choices affect inference and offer some practical recommendations for empirical researchers.

The choice of estimation window, $R_i$, involves a trade-off. In our class of DGPs, a larger $R_i$ generally yields an estimator with smaller variance, but a shorter $R_i$ can guard against violations of our assumptions stemming from parameter instability in pre-treatment data. While plotting pre-treatment time series offers informal guidance on such instability, applying our method to long-T settings reveals a heightened sensitivity to the chosen estimation window. This is because the increased likelihood of structural changes and parameter instability over extended periods becomes a critical concern. This challenge resonates with discussions in the regression discontinuity literature regarding bandwidth selection (e.g., \citealt{GelmanImbens2019}).

Regarding the choice of basis functions, Theorems \ref{th:Unbiasedness1} and \ref{th:Unbiasedness2} imply that this choice only matters for unbiasedness when the DGP has a deterministic time trend, in which case the basis functions need to be correctly specified to ensure unbiasedness (up to the order). When the DGP is mean stationary or has a stochastic trend, the choice of basis functions does not matter for unbiasedness. Basis functions may be chosen based on the time series properties of pre-treatment outcomes. Polynomial time trends seem to be a natural choice of basis functions for DGPs with deterministic trends.\footnote{In applications using difference-in-differences (DiD) methods it is typical to assume the presence of time trends (mostly linear) that are common between control and treatment groups. Our results make it clear that a linear time trend can only be dealt with by either using an untreated group or, when an untreated group is not available, by using (at least two) pre-treatment time periods to model the trend (which leads to our polynomial regression).} Other basis functions could be used, e.g., periodicity could be captured by Fourier basis functions. Our practical recommendation, and what we focus on henceforth, is to consider polynomial basis functions by letting $b_k(t)=t^k$ in Definition \ref{basisdef}.\footnote{An advantage of polynomial basis functions is that it is in principle possible to relax the assumption that the basis functions are correctly specified by instead assuming that $y_{it}^{(3)}(0)$ in Assumption \ref{ass:Series2} is a continuous (but otherwise unspecified) function of time. Since time in our setting is defined on a compact interval and the deterministic trend is a continuous function, $y_{it}^{(3)}(0)$ can be approximated arbitrarily well by a polynomial in time. In fact, by the Weierstrass Approximation Theorem, the approximation error approaches zero as the order of the polynomial goes to infinity. Under additional smoothness assumptions on the deterministic trend, an approximation theorem could then be used (e.g. the Polynomial Approximation Error Theorem) to derive a bound on the approximation error of $y_{it}^{(3)}(0)$ by the polynomial regression. The forecast of $y_{i\tau+h}(0)$ is biased, but we conjecture that it may be possible to do bias correction given an expression for the bias obtained via the approximation theorem.}

Regarding the choice of the order of the basis functions $q_i$ used for the estimation, this only matters for unbiasedness if the DGP has a deterministic trend component. For DGPs without such a trend, any $q_i$ ensures unbiasedness, with $q_i=0$ offering the lowest variance; however, larger $q_i$ can mitigate bias from non-stationary initial conditions. If a deterministic trend exists, $q_i$ cannot be smaller than the unknown true order, forcing a trade-off where larger $q_i$ for unbiasedness risks higher variance. This choice is inherently constrained by the number of pre-treatment periods ($T$): larger $T$ allows for more flexible polynomial forecasts and extensive placebo testing, yet estimation consumes periods from this budget; conversely, smaller $T$ drastically limits both model flexibility and the scope for placebo tests.

Standard cross-validation methods typically target predictive accuracy, not the unbiasedness crucial for our context. While bias-targeting cross-validation is a promising avenue for selecting $q_i$, its implementation faces challenges in our short-T panel context. It inherently requires long pre-treatment series for data splitting, which is often not feasible. Moreover, even with long-T data, this can heighten concerns about parameter instability (as stationarity becomes less defensible over extended periods). Therefore, our practical recommendation is to report results for a small range of $q_i$ values (e.g., $0,1,2,3$), with pre-treatment time series plots providing informal guidance.

\subsection{Individual treatment timing and limited anticipation} 
\label{stagger}

The treatment timing $\tau_i$ can be individual-specific. While staggered adoption designs typically allow for identification via comparisons with not-yet-treated units, such strategies are inherently limited to the window before the last unit becomes treated. In settings where all units eventually adopt the treatment, such as the universal rollout of a national policy, standard event-study designs cannot recover treatment effects for the periods following full adoption, as no valid comparison group remains.

Our approach overcomes this limitation. Because FAT constructs counterfactuals using only the treated unit's own pre-treatment history, it does not rely on the existence of a contemporaneous control group. Consequently, it allows researchers to estimate treatment effects over longer horizons, extending into periods after all units have been treated.

We require that the treatment timing be exogenous with respect to the counterfactual outcome path $y_{it}(0)$, conditional on the unit's history used for forecasting. The basis functions $\{b_k(t)\}$ in Definition \ref{basisdef} are then functions of the time to adoption $t - \tau_i$. Additionally, it is possible to allow for treatment anticipation, as long as it is limited. In this case, one can simply modify the pre-treatment estimation window ${\cal T}_{i}$ in Definition \ref{basisdef} to include observations only up to the time $\tau_i-\delta_i$ at which it is still reasonable to assume that there was no treatment anticipation, that is, ${\cal T}_{i} \equiv \left\{ \tau_i-\delta_i-R_{i}+1,\ldots,\tau_i-\delta_i\right\}$ (and $h$ is adjusted accordingly).

\subsection{Balanced panels and pooled estimation} 
\label{BalancedPanel}
For balanced panels, our proposed estimator is algebraically equivalent to two simpler pooled estimation strategies. These alternative approaches are also suitable for repeated cross-section data or when cohort of birth plays the role of time. However, in unbalanced panels, these simpler strategies may yield inconsistent estimates for ATT. Therefore, our exposition primarily focuses on individual-level estimators, given their ability to accommodate unbalanced panels; these same pooled estimators, however, are how the FAT is obtained in balanced panels.

Assume that $R=R_{i}$ and $q=q_{i}$ are constant across $i$ and focus on polynomial basis functions in Definition \ref{basisdef}. The
first alternative way to obtain our proposed estimator is to consider the cross-sectional averages $\overline{y}_{t}=\frac{1}{n}\sum_{i=1}^{n}y_{it}$
of the observed outcomes in time period $t$. Due to linearity of
the forecasting procedure, we can rewrite:
\begin{align}
\widehat{{\rm FAT}}_{h} & =\overline{y}_{\tau+h}-\sum_{k=0}^{q}\overline{\alpha}_{k}\,(\tau+h)^{k}, & \overline{\alpha} & :=\argmin_{\alpha\in\mathbb{R}^{q+1}}\sum_{t\in{\cal T}}\left(\overline{y}_{t}-\sum_{k=0}^{q}\alpha_{k}\,t^{k}\right)^{2},\label{Alternative1}
\end{align}
where ${\cal T}=\{\tau-R+1,\ldots,\tau\}$, and we suppress the dependence
on $\tau$, $q$, $R$. Here, the cross-sectional averages for $t\leq\tau$
are used to obtain a forecast of the average counterfactual for $t=\tau+h$, which is then subtracted
from the cross-sectional average observed at that time.

The second alternative is to consider
a pooled regression estimator, $\widehat{{\rm FAT}}_{h}=\widehat{\beta}_{h}$,
where 
\begin{align}    \left(\widehat{\beta},\widehat{\alpha}\right) & =\argmin_{\big\{\beta\in\mathbb{R}^{h},\,\alpha\in\mathbb{R}^{n\times(q+1)}\big\}}\,\sum_{i=1}^{n}\sum_{t=\tau-R+1}^{\tau+h}\left(y_{it}-\sum_{k=1}^{h}{1}\{t=\tau+k\}\,\beta_{k}-\sum_{k=0}^{q}\alpha_{ik}\,t^{k}\right)^{2},\label{Alternative2}
\end{align}
which is the OLS estimator obtained from regressing $y_{it}$ on a
set of time dummies $1(t=\tau+k)$, for $k\in\{1,\ldots,h\}$,
and individual-specific time trends.\footnote{It actually does not matter for $\widehat{\beta}$ whether the coefficients $\alpha$ on the time trend are individual-specific.}

\section{Parametric-Model FAT}\label{sec:Covariates}

In this section, we consider an alternative approach that explicitly incorporates covariates, including lagged outcomes, in the estimation of FAT. Compared to the baseline case, this involves specifying a parametric model for the counterfactuals and imposing additional assumptions. 

\subsection{Homogeneous parameters\label{subsec:Homogenous}}

Consider the following parametric model for the counterfactual $y_{it}(0)$:
\begin{align}
y_{it}(0) & =x_{it}'\,\beta+\sum_{k=0}^{q_{i}}\,c_{ik}t^{k}+\varepsilon_{it},\label{modelChoice}
\end{align}
where $x_{it}\in\mathbb{R}^{\dim x_{it}}$ is a vector of
covariates, $\beta\in\mathbb{R}^{\dim x_{it}}$ is restricted to be homogeneous across individuals, $c_{ik} \in \mathbb{R}$ are unknown parameters and $\varepsilon_{it}\in\mathbb{R}$ is such that
\begin{align}
\mathbb{E}\left[\varepsilon_{it}\,\big|x_{it},x_{it-1},\ldots,\varepsilon_{it-1},\varepsilon_{it-2}\right]=0.\label{modelChoice2}
\end{align}

Assuming that we have estimates $\widehat{\beta}$ for the common parameter that are consistent as $n\rightarrow\infty$
under correct model specification\footnote{For example, when $q_{i}=0$ and $x_{it}=(y_{it-1},z'_{it})'$, a consistent estimator for $\beta=(\rho,\theta')'$ can be obtained by applying
an IV regression to the first-differenced model 
\begin{align*}
y_{it}-y_{it-1} & =\left[y_{it-1}-y_{it-2}\right]\,\rho+\left[z_{it}-z_{it-1}\right]'\theta+\varepsilon_{it}-\varepsilon_{it-1},
\end{align*}
using, for example, $y_{it-2}$ and $z_{it-1}$ as instruments. In the Monte Carlo simulations, we further extend this case to $q_{i}=1$.} leads to the \textit{parametric-model} FAT, defined as:
\begin{align}
    \widehat{{\rm FAT}}_{h}^{\rm MB}=\frac{1}{n}\sum_{i=1}^{n}\left[y_{i\tau+h}
     -\widehat{y}_h(\widehat{\beta},y_i,x_i)\right],\label{MB_FAT}
\end{align}
where $\widehat{y}_h(\widehat{\beta},y_i,x_i)$ is the forecast obtained as:
\begin{align}
    \widehat{y}_h(\widehat{\beta},y_i,x_i) & :=x_{i\tau+h}'\,\widehat{\beta}+\sum_{k=0}^{q_{i}}(\tau+h)^{k}\,\widehat{c}_{ik}^{(q_i,R_i)}(\widehat{\beta}),\label{eq:MB_forecast}\\
    \widehat{c}_{i}^{(q_i,R_i)}(\widehat{\beta}) & :=\argmin_{c\in\mathbb{R}^{q_i+1}}\sum_{t\in{\cal T}_i}\left(y_{it}-x_{it}'\,\widehat{\beta}-\sum_{k=0}^{q_i}t^{k}\,c_{k}\right)^{2},\label{MB_coeffs}
\end{align}
where $c_{i}=(c_{i,0},\ldots,c_{i,q_i})$ is a $q_i+1$ vector, and ${\cal T}_i=\{\tau-R_i+1,\ldots,\tau\}$
is the set of the $R_i$ time periods directly preceding the treatment
date. The parameters $q_i$ and $R_i\in\{q_i+1,\ldots,\tau\}$
are chosen by the researcher.

\begin{thm} \label{AsyNofMBFAT}
   Consider $\widehat{{\rm FAT}}_{h}^{\rm MB}$ in \eqref{MB_FAT}. Assume that
   \begin{enumerate}[(i)]
   
      \item The forecast is unbiased when evaluated at the true parameter value $\beta_0$, i.e., $$\mathbb{E}\left[\widehat{y}_h(\beta_0,y_i,x_i)
      - y_{i\tau+h}\left(0\right) \right] = 0.$$
   
      \item The function $\widehat{y}_h(\beta,y_i,x_i)$
       is twice continuously differentiable
       such that
       $\frac{\partial \widehat{y}_h(\beta_0,y_i,x_i)} {\partial \beta}$
      has finite second moments, and
       for some $\delta>0$ we have
       $$R_n := \sup_{\left\{ \beta \, : \, \left\| \beta - \beta_0 \right\| \leq \delta \right\}}
     \left\| \frac{1}{n} \sum_{i=1}^{n} 
     \frac{\partial^2 \widehat{y}_h(\beta,y_i,x_i)} {\partial \beta \partial \beta'} \right\| = o_P(n^{1/2}).$$

      \item The estimator $\widehat{\beta}$ satisfies 
     \begin{equation}
         \widehat{\beta} - \beta_0 = \frac 1 n \sum_{i=1}^n \psi(y_i,x_i) + r_n,
     \end{equation} 
      where $\psi(y_i,x_i) $ has zero mean and finite variance, and $r_n= o_P(n^{-1/2})$. 
      Together with assumption (i) this implies that
      $\widehat{\beta} - \beta_0 =O_P(n^{-1/2})$.

      \item The sequence of random variables
      \begin{equation}
          u^*_{i\tau+h} :=  y_{i\tau+h}
       -\widehat{y}_h(\beta_0,y_i,x_i)
       - \frac{1}{n}\sum_{j=1}^n\mathbb{E}\left[\frac{\partial \widehat{y}_h(\beta_0,y_j,x_j)} {\partial \beta'}\right]
        \psi(y_i,x_i) \label{ustar}
      \end{equation} 
      satisfies a CLT in the sense that
      $$\frac{\frac{1}{\sqrt{n}}\sum_i \left(u^*_{i\tau+h}-\mathbb{E}u^*_{i\tau+h}\right)}{\bar{\sigma}^{*}_n}\Rightarrow{\cal N}\left(0,1\right), \quad \bar{\sigma}^{*2}_n := {\rm Var}\left(\frac{1}{\sqrt{n}}\sum_i u^*_{i\tau+h} \right)<\infty.$$
        
   \end{enumerate}
    Then we have that $\sqrt{n}\frac{\widehat{{\rm FAT}}_{h}^{\rm MB}-{\rm ATT}_{h}}{\bar{\sigma}^{*}_n}\Rightarrow{\cal N}\left(0,1\right)$.\footnote{
Estimation of the variance $(\bar{\sigma}^{*}_n)^2$
is discussed in Appendix~\ref{app:VarEstimation}.
    }
\end{thm}

When the covariate vector includes lagged outcomes (e.g., $x_{it} = y_{i,t-1}$), the forecast in \eqref{eq:MB_forecast} requires these lagged values to be observed. For $h=1$, this poses no problem since $y_{i\tau}$ is observed. For $h \geq 2$, however, the required intermediate counterfactuals $y_{i,\tau+1}(0), \ldots, y_{i,\tau+h-1}(0)$ are not observed under treatment, so iterated forecasting is not directly feasible. One solution is to adopt a local-projection approach that directly models the $h$-step-ahead relationship between $y_{i,\tau+h}(0)$ and $x_{i\tau}$, analogous to \cite{Jorda2005}. Alternatively, one may use the Unobserved-Components FAT from Section~\ref{sec:Baseline-model}, which does not require dynamic covariates and remains feasible at all horizons.

\subsection{Heterogeneous parameters}

If the model for the counterfactuals is an AR(p) with heterogeneous parameters, the time series literature (e.g., \citealt{FullerHasza,Dufour1984}) has derived conditions
under which forecasts from an individual AR(p) model are unbiased. The
assumptions are stationarity of the initial condition and symmetry
of the error term. 

A second example is strictly exogenous covariates with heterogeneous coefficients. For example, suppose that the $h=1$ period forecast of $y_{i\tau+1}\left(0\right)$ is given by 
\begin{align}
    &  & \widehat{y}_{i\tau+1}^{\left(q_{i},R_{i}\right)} & :=\sum_{k=0}^{q_{i}}\widehat{c}_{k}^{\left(q_{i},R_{i}\right)}\,\left(\tau+1\right){}^{k}+\hat{\beta}^{\left(i\right)}x_{i\tau+1},\label{eq:UC_forecast-1} \\
    &  & \left(\widehat{c}^{\left(q_{i},R_{i}\right)},\hat{\beta}^{\left(i\right)}\right) & :=\argmin_{\alpha\in\mathbb{R}^{q_{i}+1},\beta\in\mathbb{R}^{x}}\sum_{t\in{\cal T}_{i}}\left(y_{it}-\sum_{k=0}^{q_{i}}c_{k}\,t^{k}-\beta x_{it}\right)^{2}.\label{eq:UC_coeffs-1}
\end{align}
The forecast \eqref{eq:UC_forecast-1} is unbiased provided that a Vandermonde
matrix which includes functions of $\left(\tau-R_{i}+s\right)^{j},j=0,\dots,q_{i},s=1,2,\dots,R_{i}$
and the covariates $x_{i\tau-R_{i}+s},s=1,2,\dots,R_{i},$ is
invertible. This condition constrains how
the covariates can change over time.

\section{Simulation studies}
\label{sec:mc}

In this section, we study the finite-sample performance of our estimators. Although we focus on a universal-treatment setting, the insights from the simulations carry over unchanged to staggered adoption, as linearity of the forecasting operator ensures unbiased individual forecasts aggregate to unbiased cohort forecasts.

In Section~\ref{subsec:mc-pr-mb}, we compare the Unobserved-Components FAT (UC FAT) to the Parametric-Model FAT (PM FAT) across forecast horizons, sample sizes, autoregressive persistence, and alternative initial-condition regimes (stationary vs.\ nonstationary). PM FAT can inherit small-sample biases from first-stage estimation of the AR parameter, whereas UC FAT avoids these by construction; the contrast is most pronounced under high persistence and small $n$. When initial conditions are stationary, finite-sample biases are generally small across forecast horizons.

Section~\ref{indicator} focuses on UC FAT and evaluates how tuning parameter choices - polynomial order $q$ and pre-treatment window $R$ - trade off bias and variance under DGPs satisfying Assumptions~\ref{ass:Series} and~\ref{ass:Series2}. In the small-$T$ setting we consider, the simulations support the recommendation to use the lowest polynomial order that accommodates any deterministic trend in the outcome and, conditional on that choice, select the largest feasible pre-treatment window. When the trend order is underfit, bias emerges; in that case, shortening the pre-treatment window can mitigate misspecification. Once the polynomial order is rich enough, bias is largely insensitive to $(q,R)$, and precision improves with lower order and longer windows. These patterns are robust across DGPs with heterogeneous autoregressive parameters and individual-specific trends, and they provide clear guidance for empirical implementation.

\subsection{Performance across forecast horizons}\label{subsec:mc-pr-mb}

We consider the following DGP for the counterfactual outcome:
\[
y_{it}(0) \;=\; y^{(1)}_{it}(0) + y^{(3)}_{it}(0), \qquad t=1,\dots,T,
\]
where the autoregressive component has a homogeneous AR parameter
\[
y^{(1)}_{it}(0) = \mu_i + \rho\, y^{(1)}_{i,t-1}(0) + u_{it}, \quad
\mu_i \sim \mathcal U[-1,1],\quad u_{it}\sim\mathcal N(0,1),\quad \rho\in\{0.2,0.9\},
\]
and the deterministic trend is homogeneous linear
\[
y^{(3)}_{it}(0) = \delta_i t,\qquad \delta_i=1.
\]
We consider two regimes for $y^{(1)}_{i0}(0)$:
\begin{align*}
y^{(1)}_{i0}(0)&\sim\mathcal N(1,2) &\text{(nonstationary)},\\ 
y^{(1)}_{i0}(0)&\sim\mathcal N\!\big(\tfrac{\mu_i}{1-\rho},\,\tfrac{1}{1-\rho^2}\big)&\text{(stationary)}.
\end{align*}

We work with a balanced panel, $T=8$ and $\tau=5$ (five pre-treatment periods), forecast horizons $h\in\{1,2,3\}$, and sample sizes $n\in\{50,1000\}$. Treatment effects at $\tau+h$ are zero by construction.

\paragraph{Estimators.}
For $b_k(t)=t^k$, the UC forecast for unit $i$ at $\tau+h$ is
\[
\widehat y^{\text{UC}}_{i,\tau+h}(0)=\sum_{k=0}^{q} \widehat c_{ik}(\tau+h)^k,\quad
\widehat c_i=\arg\min_{c\in\mathbb R^{q+1}}\sum_{t=1}^{\tau}\Big(y_{it}-\sum_{k=0}^q c_k t^k\Big)^2,
\]
and our UC FAT estimator is
\begin{equation}
\label{eq:UC_FAT}
\widehat{\text{FAT}}^{\text{UC}}_h=\frac1n\sum_{i=1}^n \big(y_{i,\tau+h}-\widehat y^{\text{UC}}_{i,\tau+h}(0)\big).
\end{equation}
We compute \eqref{eq:UC_FAT} using the \texttt{fatEstimator} package in R, with $R=q+1$ and $q\in\{0,1,2,3\}$. 

The PM forecast first estimates $\rho$ by Anderson–Hsiao using $y_{i,t-2}$ as an instrument for $\Delta y_{i,t-1}$ in $\Delta y_{it}=\rho\,\Delta y_{i,t-1}+\Delta u_{it}$, then fits a polynomial to the residuals $y_{it}-\widehat\rho\,y_{i,t-1}$, so that:
\[
\widehat y^{\text{PM}}_{i,\tau+h}(0)=\widehat\rho\,y_{i,\tau}+\sum_{k=0}^{q}\widehat c_{ik}(\tau+h)^k,\quad
\widehat c_i=\arg\min_{c}\sum_{t=1}^{\tau}\Big(y_{it}-\widehat\rho\,y_{i,t-1}-\sum_{k=0}^{q} c_k t^k\Big)^2,
\]
with the PM FAT estimator given by:
\begin{equation}
\label{eq:MB_FAT}
\widehat{\text{FAT}}^{\text{PM}}_h=\frac1n\sum_{i=1}^n\big(y_{i,\tau+h}-\widehat y^{\text{PM}}_{i,\tau+h}(0)\big).
\end{equation}
We estimate $\rho$ using the \texttt{ivreg} package in R and then use the \texttt{fatEstimator} package on the residuals $y_{it}-\widehat\rho\,y_{i,t-1}$. The specification fits polynomials of degree $q\in\{0,1,2,3\}$ to $R=q+1$ pre-treatment periods.

\paragraph{Simulation scenarios.}
We consider four scenarios, varying the cross-sectional size and the autoregressive parameter:  
\[
(n,\rho) \in \{(50,0.2),\ (50,0.9),\ (1000,0.2),\ (1000,0.9)\}.
\]  

We report bias and RMSE for UC and PM FAT across polynomial orders $q\in\{0,1,2,3\}$, forecast horizons $h\in\{1,2,3\}$, AR parameter magnitude $\rho\in\{0.2,0.9\}$, and sample size $n\in\{50,1000\}$, based on 500 Monte Carlo replications.

Tables \ref{tab:UCFAT-nonstationary-h} and \ref{tab:UCFAT-stationary-h} report bias and RMSE for the UC FAT for a DGP with initial condition drawn from the nonstationary distribution (Table \ref{tab:UCFAT-nonstationary-h}) and from the stationary distribution (Table \ref{tab:UCFAT-stationary-h}). With a nonstationary initial condition (Table~\ref{tab:UCFAT-nonstationary-h}), the $q=0$ specification - which omits the linear trend - exhibits large upward bias that grows roughly linearly with the forecast horizon $h$. Including a linear term ($q=1$) effectively removes this bias: the remaining deviations from zero are small across all horizons and values of $\rho$. Higher-order specifications ($q=2,3$) do not reduce asymptotic bias in this DGP (the true trend is linear), but they affect small-sample behavior: they can attenuate some initialization effects at the cost of substantially larger RMSE because of polynomial extrapolation. In all cases, RMSE increases with the forecast horizon.

With a stationary initial condition (Table~\ref{tab:UCFAT-stationary-h}), the $q=0$ specification still omits the linear time trend and exhibits a systematic upward bias. Once a linear term is included ($q=1$), the forecast specification matches the deterministic component of the DGP, and bias is essentially zero across horizons and values of $\rho$; the small nonzero entries in the table reflect finite-sample variability. Higher-order specifications ($q=2,3$) remain asymptotically unbiased but exhibit larger RMSE due to higher-order polynomial extrapolation. As expected, RMSE increases with $h$, while bias patterns are stable across $\rho$.

Tables \ref{tab:MBFAT_nonstationary_h} and \ref{tab:MB-h-stationary} report bias and RMSE of PM FAT for a DGP with initial condition drawn from the nonstationary distribution (Table \ref{tab:MBFAT_nonstationary_h}) and from the stationary distribution (Table \ref{tab:MB-h-stationary}). With a nonstationary initial condition (Table~\ref{tab:MBFAT_nonstationary_h}), the $q=0$ specification again omits the linear trend and generates a large upward bias that increases with the forecast horizon $h$. Introducing a linear term ($q=1$) removes this trend-misspecification bias at the population level, so that any remaining bias is driven by first-stage estimation of the AR parameter $\rho$. In finite samples, PM FAT inherits a small Nickell-type bias through $\hat{\rho}$, which propagates into the $h$-step forecast and yields modest biases even when $q=1$. Higher-order specifications ($q=2,3$) do not further reduce this asymptotic bias but amplify RMSE because polynomial extrapolation increases forecast variance, especially at longer horizons. As with UC FAT, RMSE rises with $h$, while the qualitative bias pattern is similar across values of $\rho$.

With a stationary initial condition (Table~\ref{tab:MB-h-stationary}), the $q=0$ model still omits the linear trend and exhibits a systematic upward bias, though somewhat attenuated relative to the nonstationary case. Once a linear term is included ($q=1$), the forecasting specification matches the deterministic component of the DGP, and PM FAT is asymptotically unbiased; the small residual biases reflect the Nickell-type distortion from estimating $\rho$ in the first step. As before, higher-order specifications ($q=2,3$) remain asymptotically unbiased but exhibit larger RMSE due to variance inflation from higher-order polynomials. Compared to the nonstationary case, the stationary initial condition reduces the magnitude of the finite-sample biases, while RMSE continues to increase with $h$ and bias patterns remain stable across $\rho$.

\begin{table}[H]
\centering
\begin{minipage}{1\textwidth}
\centering
\begin{threeparttable}
\caption{Bias and RMSE for Unobserved-Components FAT across forecast horizons and polynomial orders using $R=q+1$ pre-treatment periods. The DGP features a nonstationary initial condition and a homogeneous AR parameter.
}
\label{tab:UCFAT-nonstationary-h}
\scriptsize
\setlength{\tabcolsep}{3pt}
\renewcommand{\arraystretch}{0.9}
\begin{tabular}{
  c c c
  *{2}{S}
  *{2}{S}
  *{2}{S}
  *{2}{S}
}
\toprule
\multicolumn{3}{c}{} &
\multicolumn{2}{c}{$q=0$} &
\multicolumn{2}{c}{$q=1$} &
\multicolumn{2}{c}{$q=2$} &
\multicolumn{2}{c}{$q=3$} \\
\cmidrule(lr){4-5}\cmidrule(lr){6-7}\cmidrule(lr){8-9}\cmidrule(lr){10-11}
{$n$} & {$\rho$} & {$h$} & {Bias} & {RMSE} & {Bias} & {RMSE} & {Bias} & {RMSE} & {Bias} & {RMSE} \\
\midrule
\multirow{3}{*}{50} & \multirow{3}{*}{0.2}
& 1 &  2.0036 &  2.0133 &  0.0165 &  0.4645 & -0.0746 &  1.2652 &  0.3312 &  3.3415 \\
& & 2 &  2.9911 &  2.9991 &  0.0066 &  0.6930 & -0.1439 &  2.4915 &  0.9446 &  8.0830 \\
& & 3 &  3.9955 &  4.0011 & -0.0090 &  0.8508 & -0.4136 &  3.8825 &  1.0853 & 14.6870 \\
\addlinespace
\multirow{3}{*}{50} & \multirow{3}{*}{0.9}
& 1 &  1.8777 &  1.8906 &  0.0316 &  0.3525 &  0.0029 &  0.8469 & -0.0107 &  2.1449 \\
& & 2 &  2.8088 &  2.8229 &  0.0338 &  0.5405 & -0.0221 &  1.6945 &  0.0071 &  5.2861 \\
& & 3 &  3.7749 &  3.7879 &  0.0435 &  0.6605 & -0.1310 &  2.6204 & -0.4783 &  9.7770 \\
\addlinespace
\multirow{3}{*}{1000} & \multirow{3}{*}{0.2}
& 1 &  1.9964 &  1.9969 &  0.0029 &  0.1056 & -0.0901 &  0.2949 &  0.2897 &  0.7752 \\
& & 2 &  2.9996 &  2.9999 &  0.0228 &  0.1512 & -0.1186 &  0.5365 &  0.9476 &  1.8948 \\
& & 3 &  3.9981 &  3.9984 &  0.0138 &  0.1849 & -0.2710 &  0.8548 &  1.7317 &  3.5273 \\
\addlinespace
\multirow{3}{*}{1000} & \multirow{3}{*}{0.9}
& 1 &  1.8730 &  1.8736 &  0.0149 &  0.0808 & -0.0207 &  0.1926 & -0.0412 &  0.4738 \\
& & 2 &  2.8187 &  2.8194 &  0.0443 &  0.1229 &  0.0039 &  0.3625 &  0.0390 &  1.0985 \\
& & 3 &  3.7728 &  3.7735 &  0.0635 &  0.1613 & -0.0456 &  0.5631 & -0.0826 &  2.0449 \\
\bottomrule
\end{tabular}
\begin{tablenotes}[para,flushleft]
\footnotesize
  Results are based on $500$ Monte Carlo replications of the DGP described in Section \ref{subsec:mc-pr-mb}. For $q=0$, bias grows linearly in $h$ and is large. Once a linear trend is included ($q=1$), bias is substantially reduced and remains small across horizons. Higher-order specifications ($q=2,3$) do not remove additional bias in this DGP but inflate forecast variability.
\end{tablenotes}
\end{threeparttable}
\end{minipage}
\end{table}

\begin{table}[H]
\centering
\begin{minipage}{1\textwidth}
\centering
\begin{threeparttable}
\caption{Bias and RMSE for Unobserved-Components FAT across forecast horizons and polynomial orders using $R=q+1$ pre-treatment periods. Stationary initial condition and homogeneous AR parameter.}
\label{tab:UCFAT-stationary-h}
\scriptsize
\setlength{\tabcolsep}{3pt}
\renewcommand{\arraystretch}{0.9}
\begin{tabular}{
  c c c
  *{2}{S}
  *{2}{S}
  *{2}{S}
  *{2}{S}
}
\toprule
\multicolumn{3}{c}{} &
\multicolumn{2}{c}{$q=0$} &
\multicolumn{2}{c}{$q=1$} &
\multicolumn{2}{c}{$q=2$} &
\multicolumn{2}{c}{$q=3$} \\
\cmidrule(lr){4-5}\cmidrule(lr){6-7}\cmidrule(lr){8-9}\cmidrule(lr){10-11}
{$n$} & {$\rho$} & {$h$} & {Bias} & {RMSE} & {Bias} & {RMSE} & {Bias} & {RMSE} & {Bias} & {RMSE} \\
\midrule
\multirow{3}{*}{50} & \multirow{3}{*}{0.2}
& 1 &  1.9829 &  1.9928 & -0.0598 &  0.4593 & -0.1349 &  1.2571 & -0.2328 &  3.2557 \\
& & 2 &  2.9943 &  3.0014 & -0.0196 &  0.6481 & -0.0366 &  2.2848 & -0.0147 &  7.3752 \\
& & 3 &  4.0111 &  4.0164 &  0.0473 &  0.8165 &  0.1694 &  3.6622 &  0.3981 & 14.2240 \\
\addlinespace
\multirow{3}{*}{50} & \multirow{3}{*}{0.9}
& 1 &  1.9962 &  2.0072 & -0.0394 &  0.3575 & -0.1069 &  0.8477 & -0.2077 &  2.1326 \\
& & 2 &  2.9950 &  3.0050 & -0.0163 &  0.5070 & -0.0422 &  1.5812 & -0.0577 &  4.8478 \\
& & 3 &  4.0120 &  4.0218 &  0.0422 &  0.6725 &  0.1414 &  2.5028 &  0.3873 &  9.3314 \\
\addlinespace
\multirow{3}{*}{1000} & \multirow{3}{*}{0.2}
& 1 &  2.0005 &  2.0010 & -0.0002 &  0.1117 & -0.0032 &  0.2971 & -0.0147 &  0.7600 \\
& & 2 &  2.9976 &  2.9979 & -0.0038 &  0.1482 & -0.0216 &  0.5316 & -0.0883 &  1.6581 \\
& & 3 &  3.9998 &  4.0000 & -0.0127 &  0.1847 & -0.0576 &  0.8121 & -0.2343 &  3.0834 \\
\addlinespace
\multirow{3}{*}{1000} & \multirow{3}{*}{0.9}
& 1 &  2.0013 &  2.0017 &  0.0008 &  0.0832 & -0.0007 &  0.2051 & -0.0066 &  0.5019 \\
& & 2 &  2.9980 &  2.9985 & -0.0029 &  0.1125 & -0.0091 &  0.3661 & -0.0492 &  1.1188 \\
& & 3 &  4.0034 &  4.0039 & -0.0073 &  0.1460 & -0.0420 &  0.5682 & -0.1535 &  2.0527 \\
\bottomrule
\end{tabular}
\begin{tablenotes}[para,flushleft]
\footnotesize
Results are based on $500$ Monte Carlo replications of the DGP described in Section \ref{subsec:mc-pr-mb}. For $q=0$, the estimator omits the linear time trend and produces a systematic upward bias. Once a linear term is included ($q=1$), the forecast specification matches the deterministic component of the DGP, so bias vanishes; small deviations in the table reflect finite-sample variability. Higher-order specifications ($q=2,3$) remain unbiased asymptotically in $n$ but yield larger RMSE because polynomial extrapolation amplifies forecast variance.

\end{tablenotes}
\end{threeparttable}
\end{minipage}
\end{table}

\begin{table}[H]
\centering
\begin{minipage}{1\textwidth}
\centering
\begin{threeparttable}
\caption{Bias and RMSE for Parametric-Model FAT across forecast horizons and polynomial orders using $R=q+1$ pre-treatment periods. Nonstationary initial condition and homogeneous AR parameter.}
\label{tab:MBFAT_nonstationary_h}
\scriptsize
\setlength{\tabcolsep}{3pt}
\renewcommand{\arraystretch}{0.9}
\begin{tabular}{
  c c c
  *{2}{S}
  *{2}{S}
  *{2}{S}
  *{2}{S}
}
\toprule
\multicolumn{3}{c}{} &
\multicolumn{2}{c}{$q=0$} &
\multicolumn{2}{c}{$q=1$} &
\multicolumn{2}{c}{$q=2$} &
\multicolumn{2}{c}{$q=3$} \\
\cmidrule(lr){4-5}\cmidrule(lr){6-7}\cmidrule(lr){8-9}\cmidrule(lr){10-11}
{$n$} & {$\rho$} & {$h$} & {Bias} & {RMSE} & {Bias} & {RMSE} & {Bias} & {RMSE} & {Bias} & {RMSE} \\
\midrule
\multirow{3}{*}{50} & \multirow{3}{*}{0.2}
& 1 &  1.5740 &  1.6295 &  0.0037 &  0.5433 & -0.0259 &  1.5487 & -0.1559 &  4.0561 \\
& & 2 &  2.3543 &  2.4345 & -0.0251 &  0.7298 & -0.2214 &  2.7906 & -0.8179 &  8.9927 \\
& & 3 &  3.1187 &  3.2246 &  0.0497 &  0.9475 &  0.1381 &  4.1904 & -0.1103 & 15.8752 \\
\addlinespace
\multirow{3}{*}{50} & \multirow{3}{*}{0.9}
& 1 &  0.8352 &  1.4225 & -0.0538 &  0.5179 & -0.0828 &  1.4012 & -0.1100 &  3.5624 \\
& & 2 &  1.3678 &  2.1894 & -0.0508 &  0.6659 & -0.2010 &  2.4168 & -0.5022 &  7.7318 \\
& & 3 &  1.7103 &  2.8662 & -0.0006 &  0.8940 &  0.1780 &  3.7222 &  0.8174 & 13.8606 \\
\addlinespace
\multirow{3}{*}{1000} & \multirow{3}{*}{0.2}
& 1 &  1.5999 &  1.6028 &  0.0107 &  0.1181 &  0.0264 &  0.3309 &  0.0428 &  0.8568 \\
& & 2 &  2.4025 &  2.4061 &  0.0012 &  0.1462 & -0.0091 &  0.5638 & -0.0176 &  1.9492 \\
& & 3 &  3.1996 &  3.2046 &  0.0028 &  0.2096 &  0.0177 &  0.9540 &  0.0418 &  3.5432 \\
\addlinespace
\multirow{3}{*}{1000} & \multirow{3}{*}{0.9}
& 1 &  0.2371 &  0.3189 &  0.0084 &  0.1166 &  0.0243 &  0.3281 &  0.0519 &  0.8504 \\
& & 2 &  0.3438 &  0.4728 & -0.0018 &  0.1464 & -0.0059 &  0.5634 &  0.0306 &  1.9445 \\
& & 3 &  0.4719 &  0.6387 & -0.0014 &  0.2074 &  0.0171 &  0.9476 &  0.1110 &  3.5020 \\
\bottomrule
\end{tabular}
\begin{tablenotes}[para,flushleft]
\footnotesize
Results are based on $500$ Monte Carlo replications of the DGP described in Section \ref{subsec:mc-pr-mb}. For $q=0$, the estimator omits the linear time trend and produces a systematic upward bias. Once a linear term is included ($q=1$), the bias disappears. Higher-order specifications ($q=2,3$) remain unbiased asymptotically in $n$ but yield larger RMSE. RMSE grows with $h$, while bias patterns remain stable across different values of~$\rho$.
\end{tablenotes}
\end{threeparttable}
\end{minipage}
\end{table}

\begin{table}[H]
\centering
\begin{minipage}{1\textwidth}
\centering
\begin{threeparttable}
\caption{Monte Carlo Results: Bias and RMSE for Parametric-Model FAT across forecast horizons and polynomial using $R=q+1$ pre-treatment periods. Stationary initial condition and homogeneous AR parameter.}
\label{tab:MB-h-stationary}
\scriptsize
\setlength{\tabcolsep}{3pt}
\renewcommand{\arraystretch}{0.9}
\begin{tabular}{
  c c c
  *{2}{S}
  *{2}{S}
  *{2}{S}
  *{2}{S}
}
\toprule
\multicolumn{3}{c}{} &
\multicolumn{2}{c}{$q=0$} &
\multicolumn{2}{c}{$q=1$} &
\multicolumn{2}{c}{$q=2$} &
\multicolumn{2}{c}{$q=3$} \\
\cmidrule(lr){4-5}\cmidrule(lr){6-7}\cmidrule(lr){8-9}\cmidrule(lr){10-11}
{$n$} & {$\rho$} & {$h$} & {Bias} & {RMSE} & {Bias} & {RMSE} & {Bias} & {RMSE} & {Bias} & {RMSE} \\
\midrule
\multirow{3}{*}{50} & \multirow{3}{*}{0.2}
& 1 &  1.6119 &  1.6686 &  0.0154 &  0.5343 & -0.0164 &  1.4927 & -0.1151 &  3.9297 \\
& & 2 &  2.4122 &  2.4798 &  0.0148 &  0.7293 & -0.1584 &  2.6326 & -0.7999 &  8.4156 \\
& & 3 &  3.1723 &  3.2695 &  0.0709 &  0.9797 &  0.2025 &  4.5242 &  0.1281 & 17.0235 \\
\addlinespace
\multirow{3}{*}{50} & \multirow{3}{*}{0.9}
& 1 &  0.7370 &  1.2933 &  0.0108 &  0.4967 &  0.0084 &  1.3603 &  0.0136 &  3.5382 \\
& & 2 &  1.1573 &  1.9940 &  0.0348 &  0.6901 & -0.0844 &  2.4447 & -0.5345 &  7.7359 \\
& & 3 &  1.6071 &  2.7400 &  0.1014 &  0.9006 &  0.3299 &  4.0615 &  0.7651 & 15.2754 \\
\addlinespace
\multirow{3}{*}{1000} & \multirow{3}{*}{0.2}
& 1 &  1.6071 &  1.6097 &  0.0142 &  0.1143 &  0.0309 &  0.3139 &  0.0634 &  0.8297 \\
& & 2 &  2.3981 &  2.4016 &  0.0087 &  0.1664 &  0.0206 &  0.5888 &  0.0087 &  1.9239 \\
& & 3 &  3.1944 &  3.1981 & -0.0003 &  0.2023 &  0.0102 &  0.9322 &  0.0347 &  3.6121 \\
\addlinespace
\multirow{3}{*}{1000} & \multirow{3}{*}{0.9}
& 1 &  0.2328 &  0.3175 &  0.0145 &  0.1136 &  0.0327 &  0.3122 &  0.0680 &  0.8272 \\
& & 2 &  0.3432 &  0.4814 &  0.0090 &  0.1653 &  0.0231 &  0.5867 &  0.0129 &  1.9199 \\
& & 3 &  0.4503 &  0.6249 &  0.0014 &  0.2016 &  0.0165 &  0.9265 &  0.0547 &  3.5848 \\
\bottomrule
\end{tabular}
\begin{tablenotes}[para,flushleft]
\footnotesize
With stationary initial conditions, the $q=0$ specification still produces upward bias from omitting the linear trend. Once a linear term is included ($q=1$), the PM estimator is asymptotically unbiased; the small residual biases reflect estimation of $\rho$ in the first step. Higher-order $q$ remain unbiased but exhibit larger RMSE due to variance inflation. Compared to the nonstationary case, the stationary initial condition attenuates the size of the finite-sample biases.
\end{tablenotes}
\end{threeparttable}
\end{minipage}
\end{table}

\subsection{Choice of tuning parameters for Unobserved-Components FAT} \label{indicator}

In this section, we compare the finite-sample performance of the Unobserved-Components FAT estimator across different tuning parameters: the polynomial order, $q$, and the estimation window, $R$, for different specifications of the DGP. All specifications satisfy Assumption \ref{ass:Series2}, with the counterfactual process specified as the sum of up to three different components. 

Let $I_j, j=1,2,3$ be an indicator that equals to one whenever the associated components is present in the specification of the counterfactual process. For each $i=1,\dots,n$: 
\begin{align*}
    y_{it}(0) &= I_1{y}_{it}^{(1)}(0) +I_2{y}_{it}^{(2)}(0)+I_3{y}_{it}^{(3)}(0), \; t=1,\dots,T,\\
    {y}_{it}^{(1)}(0) & =\mu_{i}+\rho {y}_{it-1}^{(1)}(0)+u_{it}, \; t\geq 1, \qquad  
    {y}_{i0}^{(1)}(0) \sim\mathcal{N}\left(\frac{\mu_i}{1-\rho},\frac{1}{1-\rho^2}\right), \\
    {y}_{it}^{(2)}(0) &= {y}_{it-1}^{(2)}(0) + \epsilon_{it}, \; t\geq 1, 
    \qquad {y}_{i0}^{(2)}(0) = 0, \\ 
    {y}_{it}^{(3)}(0) &=\delta_{i}t,\\
    \mu_{i} &\sim\mathcal{U}\left[-1,1\right], 
    \quad u_{it} \sim\mathcal{N}\left(0,1\right), 
    \quad \epsilon_{it} \sim \mathcal{N}(0,1),\\
    \rho &= 0.2, \quad \delta_{i} =1, \quad T=6, \quad \tau=5, \quad h=1, \quad n=1000.
\end{align*}

Note that the initial observation ${y}_{i0}^{(1)}$ is drawn from the stationary distribution and the time trend is linear and homogeneous, so it can be interpreted as a common shock.

Table \ref{tab:table2} shows results for the bias and standard error of $\widehat{\rm FAT}^{\rm UC}_1$ across different tuning parameters. The table shows that, when the DGP is mean stationary (first panel) or when it is the sum of a mean stationary and a random walk (second panel), the bias does not vary much across different values of the tuning parameters, while the standard errors are smaller for smaller $q$ and larger $R$. When the DGP contains a time trend component, we observe bias when the polynomial-order $q$ is less than the true order of the time trend, as the theory predicts. In this case, however, a smaller estimation window $R$ gives a smaller bias. When $q\geq 1$, the performance of the estimator in terms of bias is again robust to the choice of tuning parameters, with smaller standard errors for smaller $q$ and larger $R$.\footnote{We observe the same behavior with a sample size of $N=30$, see Table \ref{tab:table_smallN} in the Online Appendix.}

\begin{table}[]
\spacingset{1.7}
\centering
\tiny
\begin{tabular}{@{}crlllll@{}}
\toprule
\multicolumn{1}{l}{}                                                                 & \multicolumn{1}{l}{}      & \multicolumn{5}{c}{$R$}             \\ \cmidrule(l){3-7} 
                                                                                     & \multicolumn{1}{l}{}      & $q+1$ & $q+2$ & $q+3$ & $q+4$ & $q+5$ \\ \midrule
\multirow{9}{*}{\begin{tabular}[c]{@{}c@{}}Stationary AR(1) \\ $I_1=1, I_2 = 0 = I_3$ \end{tabular}}     & \multicolumn{1}{c}{$q=0$} &       &       &       &       &       \\ \cmidrule(lr){2-2}
 & bias                      & -0.0002 & -0.0005 & -0.0003  & -0.0001 & 0.0005 \\
 & s.e.                      & 0.0397 & 0.036   & 0.0354  & 0.0346  & 0.0341  \\ \cmidrule(lr){2-2}
 & \multicolumn{1}{l}{$q=1$} &        &         &         &         &         \\ \cmidrule(lr){2-2}
 & bias                      & 0.0047 & 0.0003  & 0.0009  & 0.0008  &         \\
 & s.e.                      & 0.0709 & 0.0565  & 0.0476  & 0.0448  &         \\ \cmidrule(lr){2-2}
 & \multicolumn{1}{c}{$q=2$} &        &         &         &         &         \\ \cmidrule(lr){2-2}
 & bias                      & 0.0112 & 0.0023  & 0.0015  &         &         \\
 & s.e.                      & 0.1225 & 0.0907  & 0.0726  &         &         \\ \midrule
 \multirow{9}{*}{\begin{tabular}[c]{@{}c@{}}Stationary AR(1) \\ + unit root \\ $I_1=1=I_2, I_3=0$ \end{tabular}}     & \multicolumn{1}{c}{$q=0$} &       &       &       &       &       \\ \cmidrule(lr){2-2}
 & bias                      & -0.0029 & -0.0041  & -0.0045  & -0.005  & -0.005 \\
 & s.e.                      & 0.0516 & 0.0512  & 0.0525  & 0.0547  & 0.0577  \\ \cmidrule(lr){2-2}
 & \multicolumn{1}{l}{$q=1$} &        &         &         &         &         \\ \cmidrule(lr){2-2}
 & bias                      & -0.0004 & -0.0023  & -0.0023  & -0.0033  &         \\
 & s.e.                      & 0.082 & 0.0664  & 0.0625  & 0.0606  &         \\ \cmidrule(lr){2-2}
 & \multicolumn{1}{c}{$q=2$} &        &         &         &         &         \\ \cmidrule(lr){2-2}
 & bias                      & 0.0025 & -0.0011  & -0.0002  &         &         \\
 & s.e.                      & 0.1454 & 0.0997  & 0.0868  &         &         \\ \midrule
 \multirow{9}{*}{\begin{tabular}[c]{@{}c@{}}Stationary AR(1) \\+ linear trend  \\ $I_1=1=I_3, I_2=0$ \end{tabular}} & \multicolumn{1}{l}{$q=0$} &       &       &       &       &       \\ \cmidrule(lr){2-2}
 & bias                      & 0.9998 & 1.4995 & 1.9997  & 2.4999 & 3.0005  \\
 & s.e.                      & 0.0397 & 0.036 & 0.0354  & 0.0346 & 0.0341  \\ \cmidrule(lr){2-2}
 & \multicolumn{1}{l}{$q=1$} &        &         &         &         &         \\ \cmidrule(lr){2-2}
 & bias                      & -0.0027  & -0.0024 & -0.0008 & -0.001 &         \\
 & s.e.                      & 0.068 & 0.0536  & 0.0466  & 0.0442  &         \\ \cmidrule(lr){2-2}
 & \multicolumn{1}{l}{$q=2$} &        &         &         &         &         \\ \cmidrule(lr){2-2}
 & bias                      & -0.0032 & -0.0051  & -0.0022  &         &         \\
 & s.e.                      & 0.1225 & 0.0839  & 0.0698  &         &         \\ \midrule
 \multirow{9}{*}{\begin{tabular}[c]{@{}c@{}}Stationary AR(1) \\ + linear trend \\ + unit root \\ $I_1=I_2=I_3=1$  \end{tabular}} & \multicolumn{1}{l}{$q=0$} &       &       &       &       &       \\ \cmidrule(lr){2-2}
 & bias                      & 0.9971 & 1.4959  & 1.9955  & 2.4950  & 2.9950  \\
 & s.e.                      & 0.0516 & 0.0512  & 0.0525  & 0.0547  & 0.0577  \\ \cmidrule(lr){2-2}
 & \multicolumn{1}{l}{$q=1$} &        &         &         &         &         \\ \cmidrule(lr){2-2}
 & bias                      & 0.0005 & -0.0015 & 0.001 & 0.0014 &         \\
 & s.e.                      & 0.0831 & 0.0659 & 0.0608 & 0.0621 &         \\ \cmidrule(lr){2-2}
 & \multicolumn{1}{l}{$q=2$} &        &         &         &         &         \\ \cmidrule(lr){2-2}
 & bias                      & 0.001   & -0.0047    & -0.0018  &         &         \\
 & s.e.                      & 0.1447  & 0.1024   & 0.0873  &         &         \\ \bottomrule \bottomrule
\end{tabular}
\caption{Bias and standard error (s.e.) for the Unobserved-Components FAT when the counterfactual is specified as indicated in the left-most column. Sample size $N=1000$.}
\label{tab:table2}
\end{table}
\spacingset{1.9}

\section{Empirical illustration}\label{sec:EmpiricalIllustration}

We consider an empirical exercise showing that FAT, despite not requiring an untreated control group, yields conclusions consistent with established methods that rely on a control group. This positions FAT as a valuable complementary robustness check when a control group is available, rather than a replacement for existing methods such as TWFE or DiD.\footnote{Additional replications are included in the Online Appendix \ref{sec:overdose} and \ref{sec:DIDrefugees}.} 

We consider the application in \cite{GoodmanBacon} on the effect of no-fault laws on female suicide. The data are on U.S. states that adopted no-fault divorce laws from 1969 to 1985. The outcome is the age-adjusted average suicide mortality rate per million women (ASMR) in the state. We work with a balanced panel of $n=36$ treated states and 10 time periods, 5 of which are pre-treatment.

We assume the same tuning parameters $q$ and $R$ for FAT across states. To help select the basis functions and the tuning parameters, Figure \ref{fig:suicide_adoption} plots the outcome variable averaged across states, as a function of time to adoption. The visible trending behaviour (plausibly reflecting a deterministic component) in the pre-treatment data suggests using a polynomial basis function and choosing orders $q>0$. There is no strong indication of structural instability in the pre-treatment data, so we use all available data for estimation, that is, we set $R=5$. We compute the forecast $\widehat{y}^{UC}_{i\tau+h}$ in \eqref{eq:UC_FAT} for a number of forecast horizons $h$ and orders $q$. The Unobserved-Components FAT estimates at $\tau+h$ are shown in the top panel of Figure \ref{fig:suicide_FAT_placebo} as a function of $h\in\{1,\dots,5\}$ and for $q\in\{1,2\}$. The variance of the estimator is computed as explained in Section \ref{app:VarEstimation}. The results show a statistically insignificant decrease in the suicide rate after adoption of the no-fault divorce laws across different forecast horizons. This replicates the findings in \cite{GoodmanBacon}, \cite{ClarkeTapiaSchythe}, but without using a control group. The standard error of our estimator is comparable to that obtained in the above-mentioned studies that employ event-study design approaches and, as expected, our $95\%$ confidence intervals are greater for longer forecast horizons (see, e.g., the results in Figure 2 in \citealt{ClarkeTapiaSchythe}).\footnote{To see how our estimator performs in terms of standard error with a small sample size, as in this application, see Table \ref{tab:table_smallN} in the Online Appendix for simulation results with $N=30$.}

To validate the procedure, in Figure \ref{fig:suicide_FAT_placebo} we report placebo FAT at lag $j$, $j\in\{0,\dots,3\}$, calculated as if the law was adopted $j$ years earlier than its actual adoption date (for each state). The forecast horizon is $h=1$. The figure shows that the placebo FATs for both polynomial orders are close to zero. This offers suggestive evidence that FAT estimates may be interpreted as ATT estimates.

\begin{figure}[tbh!]
\begin{centering}
\includegraphics[width=0.8\linewidth]{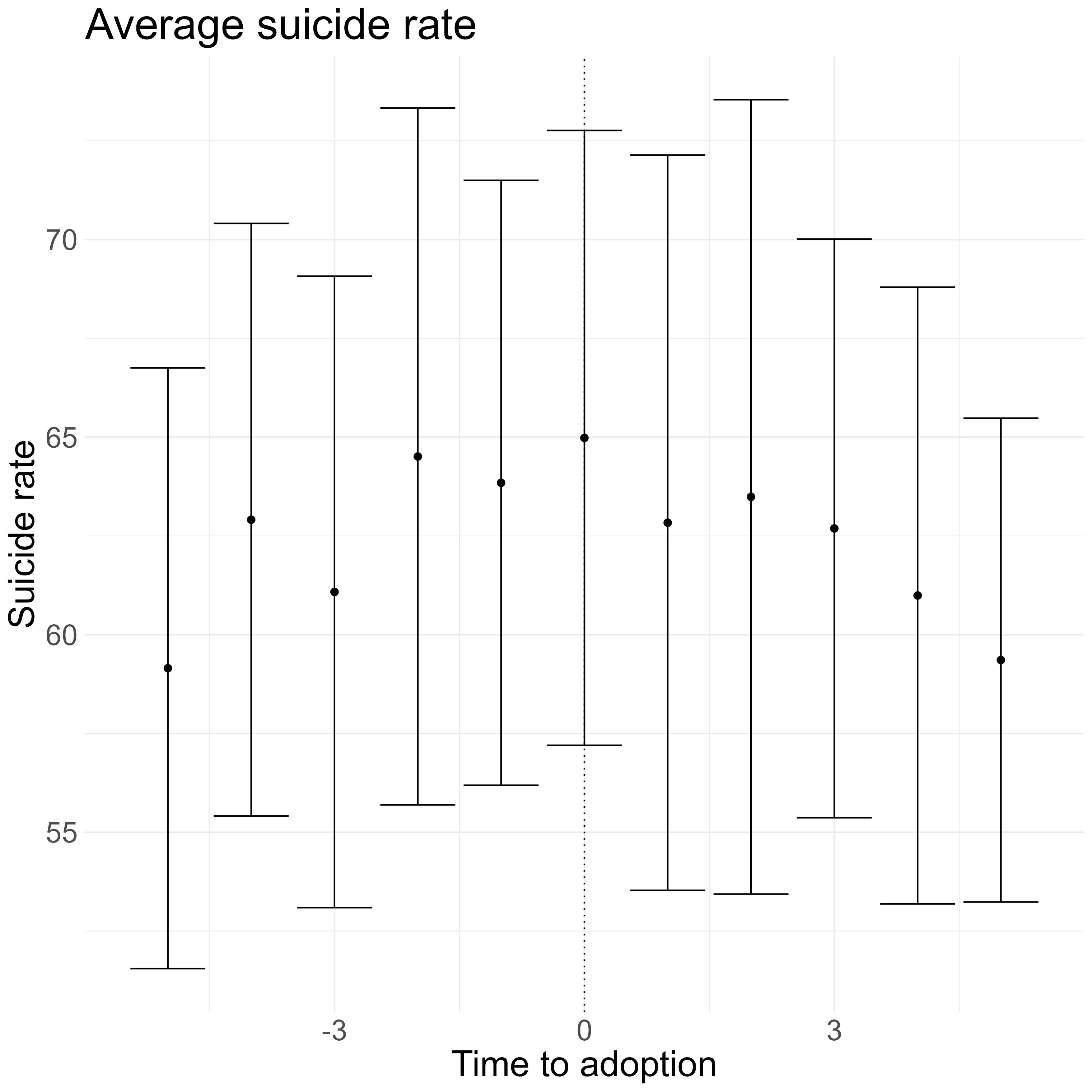}
\par\end{centering}
\caption{Average suicide rates across states as a function of time-to-adoption.}
\label{fig:suicide_adoption}
\end{figure}

\begin{figure}[tbh!]
\begin{centering}
\includegraphics[width=0.8\linewidth]{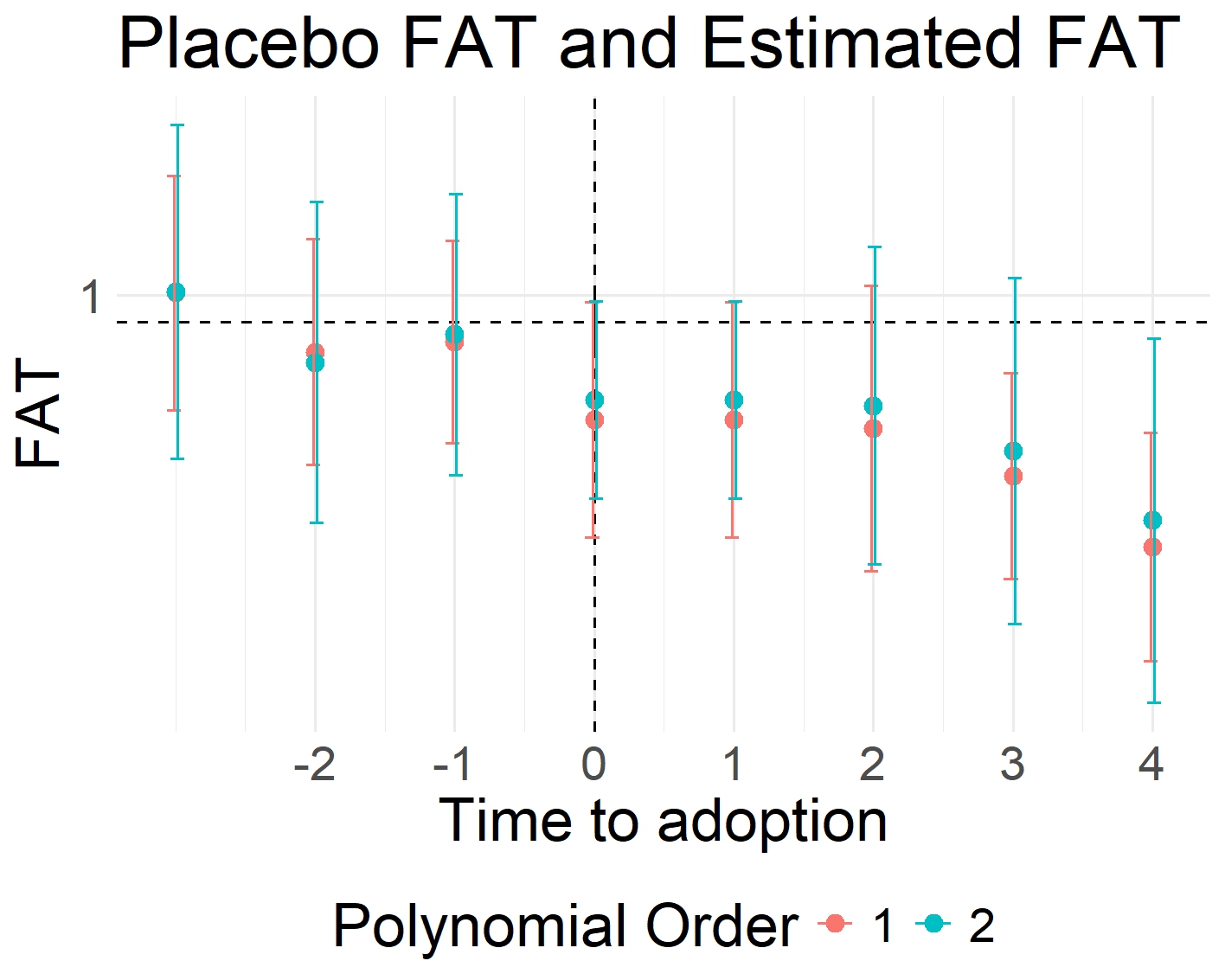}
\par\end{centering}
\caption{FAT estimates for different forecast horizons (horizontal axis) and polynomial orders (different colors). To the left of the dashed vertical line, the figure shows placebo FAT (for $h=1$) as a function of the lag (horizontal axis) and polynomial order (different lines). The error bars in both panels represent $95\%$ confidence intervals.}
\label{fig:suicide_FAT_placebo}
\end{figure}

\section{Conclusion}\label{sec:Conclusion}
This paper proposes an estimator of the average treatment effects in the absence of an untreated or control group, based on forecasting individual counterfactuals using basis function regressions over a (short) time series of pre-treatment data. Forecast unbiasedness is a key requirement that is satisfied by our approach under a broad class of DGPs that express the individual counterfactuals as the sum of up to three unobserved components: a stationary process, a stochastic trend, and a deterministic trend. The approach is robust, general, and flexible, allowing for unbalanced panels, heterogeneous treatment effects, and staggered treatment timing. Forecasting counterfactuals using a parametric model instead requires stronger assumptions and can perform poorly due to misspecification and estimation bias in small samples (e.g., incidental parameter problem).

\FloatBarrier

\spacingset{1.9}

\appendix

\section{Proofs \label{sec:Appendix_Proofs}}

\begin{proof}[Proof of Lemma \ref{lem:LemmaConsistencyFAT}]
We have 
\begin{align*}
    \widehat{{\rm FAT}}_{h}-{\rm ATT}_{h} & =\frac{1}{n}\sum_{i=1}^{n}\left(y_{i\tau+h}-\widehat{y}_{i\tau+h}\left(0\right)-\mathbb{E}\left[y_{i\tau+h}-y_{i\tau+h}\left(0\right)\right]\right)\\
    & =\frac{1}{n}\sum_{i=1}^{n}\left(y_{i\tau+h}-\widehat{y}_{i\tau+h}\left(0\right)-\mathbb{E}\left[y_{i\tau+h}-\widehat{y}_{i\tau+h}\left(0\right)\right]\right)\\
    & =\frac{1}{n}\sum_{i=1}^{n}\left(u_{i\tau+h}-\mathbb{E}u_{i\tau+h}\right),
\end{align*}
where we used Assumption \ref{Unbiasedness} to obtain the second equality above. Since our assumptions guarantee that $\left(u_{i\tau+h}-\mathbb{E}u_{i\tau+h}\right)$
has zero mean and satisfies a CLT, we obtain the desired result. 
\end{proof}
\begin{proof}[Proof of Theorem \ref{th:Unbiasedness1}]
   It is sufficient to show that for each component $y_{it}^{(r)}(0)$, $r\in\{1,2\}$, 
    \begin{align}
    \mathbb{E}\left[\sum_{t\in{\cal T}_{i}}\,w_{it}\,y_{it}^{(r)}-y_{i\tau+h}^{(r)}(0)\right] & =0.\label{unb1}
    \end{align}
    
    For both the mean stationary component ($r=1$) and the random walk component ($r=2$) we have $\mathbb{E}\left(y_{it}^{(r)}-y_{i\tau+h}^{(r)}(0)\right)=0$. Multiplying this equation by $w_{it}$, summing over $t\in{\cal T}_{i}$, and using the fact that the non-random weights $w_{it}$ sum to 1, we obtain \eqref{unb1} for $r=1$ and $r=2$. 
\end{proof}

\begin{proof}[Proof of Lemma \ref{ForecastWeights}]
    Let $R_{i}=q_{i}+1$ and $c_{s}\equiv\tau_{i}-R_{i}+s$, $s=1,2,\dots,R_{i}$. Define the $R_{i}\times\left(q_{i}+1\right)$ alternant matrix $X_{i}$ and the $1\times\left(q_{i}+1\right)$ vector $H_{i}$ as, respectively, 
    \begin{equation}
        X_{i}\equiv\left[\begin{array}{cccc}
        1 & b_{1}\left(c_{1}\right) & \dots & b_{q_{i}}\left(c_{1}\right)\\
        1 & b_{1}\left(c_{2}\right) & \dots & b_{q_{i}}\left(c_{2}\right)\\
        1 & b_{1}\left(c_{3}\right) & \dots & b_{q_{i}}\left(c_{3}\right)\\
        \dots & \dots & \dots & \dots\\
        1 & b_{1}\left(c_{R_{i}}\right) & \dots & b_{q_{i}}\left(c_{R_{i}}\right)
        \end{array}\right],\text{\;}H_{i}\equiv\left[\begin{array}{cccc}
        1 & b_{1}\left(\tau+h\right) & \dots & b_{q_{i}}\left(\tau+h\right)\end{array}\right].\label{eq:Vandy-1}
        \end{equation}
        
    The OLS coefficients from regressing $y_{i}=\left(y_{i\tau_{i}-R_{i}+1},\dots,y_{i\tau_{i}}\right)$ on $b_{k}\left(t\right)$ are given by 
    \[\widehat{c_i}^{(q_i,R_i)}=\left(X_{i}'X_{i}\right)^{-1}X_{i}'y_{i},
    \]
    so that the $R_{i}$ forecast weights are 
    \begin{align}
    w_{i}=H_{i}\left(X_{i}'X_{i}\right)^{-1}X_{i}'.\label{weights_w}
\end{align}

Since $X_{i}$ is a Vandermonde matrix with the first column being
a column of ones (by assumption), it follows that $X_{i}e_{1}=\iota.$
Then 
\[
\left(X_{i}'X_{i}\right)^{-1}X_{i}'\iota=e_{1}\equiv\left[\begin{array}{c}
1\\
0\\
\vdots\\
0
\end{array}\right],
\]
so that $w_{i}\iota=1$, where $\iota$ is the $\left(q_{i}+1\right)\times1$
vector of ones. This proves the statement.
\end{proof}

\begin{proof}[Proof of Theorem \ref{th:Unbiasedness2}]
    It is sufficient to show that for each component $y_{it}^{(r)}(0)$, $r\in\{1,2,3\}$, 
    \begin{align}
    \mathbb{E}\left[\sum_{t\in{\cal T}_{i}}\,w_{it}^{(q_{i},R_{i})}\,y_{it}^{(r)}-y_{i\tau+h}^{(r)}(0)\right] & =0.\label{ToBeShown}
    \end{align}
    
     For the mean stationary component ($r=1$) and the random walk component ($r=2$) we have $\mathbb{E}\left(y_{it}^{(r)}-y_{i\tau+h}^{(r)}(0)\right)=0$. Multiplying this equation by $w_{it}^{(q_{i},R_{i})}$, summing over $t\in{\cal T}_{i}$, and using the fact that the non-random weights sum to 1 by Lemma \ref{ForecastWeights}, we obtain \eqref{ToBeShown} for $r=1$ and $r=2$. To show \eqref{ToBeShown} for the deterministic time trend component ($r=3$), note that by \eqref{weightavg}, $\sum_{t\in{\cal T}_{i}}w_{it}^{(q_{i},R_{i})}\,y_{it}^{(3)}=\sum_{k=0}^{q_{i}}\widetilde{c}_{ik}^{(q_{i},R_{i})}\,(\tau+h)^{k}$, where 
    \begin{align*}
    \widetilde{c_i}^{(q_{i},R_{i})} & :=\argmin_{c\in\mathbb{R}^{q_{i}+1}}\sum_{t\in{\cal T}_{i}}\left(y_{it}^{(3)}-\sum_{k=0}^{q_{i}}c_{k}\,b_k(t)\right)^{2}.
    \end{align*}
    Since $q_{i}\geq q_{0i}$ for all $i$, the objective function in the last display is minimized (with value zero) at $\widetilde{c}_{ik}^{(q_{i},R_{i})}=c_{ik}^{(3)}$, which implies $y_{i\tau+h}^{(3)}(0)=\sum_{t\in{\cal T}_{i}}\,w_{it}^{(q_{i},R_{i})}\,y_{it}^{(3)}$, that is, \eqref{ToBeShown} holds for $r=3$ even without taking the expectation. 
\end{proof}
\begin{proof}[Proof of Theorem \ref{AsyNofMBFAT}]
 We have
\begin{align*}
   & \widehat{{\rm FAT}}_{h}^{\rm MB} -{\rm ATT}_{h}
   \\
   & =\frac{1}{n}\sum_{i=1}^{n}\left(y_{i\tau+h}-\widehat{y}_h(\widehat \beta,y_i,x_i)
   -\mathbb{E}\left[y_{i\tau+h}-y_{i\tau+h}\left(0\right)\right]\right)
   \\
    & =\frac{1}{n}\sum_{i=1}^{n}\left(y_{i\tau+h}
    - \widehat{y}_h(\widehat \beta,y_i,x_i) -\underbrace{
    \mathbb{E}\left[y_{i\tau+h}-\widehat{y}_h(\beta_0,y_i,x_i)\right]}_{= \mathbb{E}\left[ u_{i\tau+h}^* \right]}\right)
\\
   &= \frac{1}{n} \sum_{i=1}^{n} \left[
      y_{i\tau+h}
       -\widehat{y}_h(\beta_0,y_i,x_i)
       - \frac{\partial \widehat{y}_h(\beta_0,y_i,x_i)} {\partial \beta'}
       \left( \widehat{\beta} - \beta_0 \right)
     \right] - \mathbb{E}\left[ u_{i\tau+h}^* \right]
       \\
     &\qquad 
     + O\left( R_n
     \left\| \widehat{\beta} - \beta_0 \right\|^2 
     \right)
  \\
  &= 
    \frac{1}{n} \sum_{i=1}^{n} u_{i\tau+h}^* - \mathbb{E} \left[u_{i\tau+h}^*\right]
     + O\left( R_n
     \left\| \widehat{\beta} - \beta_0 \right\|^2 
     \right)
     \\
     &\qquad 
       - \left\{ \frac{1}{n} \sum_{i=1}^{n}
     \frac{\partial \widehat{y}_h(\beta_0,y_i,x_i)} {\partial \beta'} - \frac{1}{n}\sum_j
\mathbb{E}\left[ \frac{\partial \widehat{y}_h(\beta_0,y_j,x_j)} {\partial \beta'} \right]
     \right\}   \left( \widehat{\beta} - \beta_0 \right)
      \\
     &\qquad 
     -\frac{1}{n}\sum_j\mathbb{E}\left[ \frac{\partial \widehat{y}_h(\beta_0,y_j,x_j)} {\partial \beta'} \right]
     r_n
     \\
    &= \frac{1}{n} \sum_{i=1}^{n} u_{i\tau+h}^* - \mathbb{E} \left[u_{i\tau+h}^*\right]
      + o_P(n^{-1/2})
\end{align*}
Here, in the first step, we plugged in the definitions
of $\widehat{{\rm FAT}}_{h}^{\rm MB}$  and ${\rm ATT}_{h}$.
In the second step, we
used the unbiasedness
of the forecast, definition \eqref{ustar}, and assumption (iii) that $\mathbb{E}\left(\psi(y_i,x_i)\right)=0$.
In the third step, given assumption (ii), we employed a Taylor expansion of $\widehat{y}_h(\beta,y_i,x_i)$ in $\beta$ around $\beta_0$. In the fourth step we decomposed $\frac{\partial \widehat{y}_h(\beta_0,y_i,x_i)} {\partial \beta'}$ into its expectation and its deviation from the expectation, and used
$\widehat{\beta} - \beta_0 = \frac 1 n \sum_{i=1}^n \psi(y_i,x_i) + r_n$
and the definition of $u_{i\tau+h}^*$ in \eqref{ustar}. In the final step, we used our assumptions to conclude that the various remainder terms are all of order $o_P(n^{-1/2})$.
By an application of a standard cross-sectional CLT
we then obtain the conclusion of the theorem.
\end{proof}

\section{Estimating the variance of $\widehat{{\rm FAT}}_{h}$}
\label{app:VarEstimation}

According to Lemma~\ref{lem:LemmaConsistencyFAT} the asymptotic variance
of $\sqrt{n} \widehat{{\rm FAT}}_{h}$ is given by $\bar{\sigma}^2_n:= {\rm Var}(\frac{1}{\sqrt{n}} \sum_i u_{i\tau+h})$, which under cross-sectional
independence can be consistently estimated by
$$
    \frac 1 n \sum_{i=1}^n  \left( u_{i\tau+h} - \overline u_{\tau+h} \right)^2 ,
$$
where $ \overline u_{\tau+h} =     \frac 1 n \sum_{i=1}^n  u_{i\tau+h} $, with $u_{i\tau+h}:=y_{i\tau+h}-\widehat{y}_{i\tau+h}(0)$.

If an additional unknown common parameter $\beta_0$ needs to be estimated, the uncertainty of the estimator $\widehat \beta$ also becomes relevant. In that case, according to Theorem~\ref{AsyNofMBFAT}, the asymptotic variance of  $\sqrt{n} \widehat{{\rm FAT}}_{h}$ 
is $\bar{\sigma}^{*2}_n := {\rm Var}\left(\frac{1}{\sqrt{n}}\sum_i u^*_{i\tau+h} \right)$, where 
 $u^*_{i\tau+h} :=  y_{i\tau+h}
       -\widehat{y}_h(\beta_0,y_i,x_i)
       - \frac{1}{n}\sum_{j=1}^n\mathbb{E}\left[\frac{\partial \widehat{y}_h(\beta_0,y_j,x_j)} {\partial \beta'}\right]
        \psi(y_i,x_i) $.
Here, both $\beta_0$ and the influence function   $ \psi(y_i,x_i)$ of $\widehat \beta$ are unknown. For example, if $\widehat \beta$
is a method of moment estimator (which includes OLS, MLE, and exactly identified IV) that solves the sample moment condition
$$
     \sum_{i=1}^n  g(y_i,x_i,\widehat \beta) = 0 ,
$$
then, assuming that $g(y_i,x_i,\beta)$ is sufficiently often differentiable in $\beta$, we have
$$
   \psi(y_i, x_i) = \left( \frac{1}{n} \sum_{j=1}^n \frac{\partial^2 g(y_j, x_j,\beta_0)}{\partial \beta \partial \beta'} \right)^{-1} \frac{\partial g(y_i, x_i,\beta_0)}{\partial \beta},
$$
which can be estimated by
$$
    \widehat \psi(\widehat \beta, y_i,x_i)  
    = \left( \frac{1}{n} \sum_{j=1}^n \frac{\partial^2 g(y_j, x_j, \widehat{\beta})}{\partial \beta \partial \beta'} \right)^{-1} \frac{\partial g(y_i, x_i, \widehat{\beta})}{\partial \beta}.
$$
In that case, we can estimate $u^*_{i\tau+h} $ by
$$
  \widehat u^*_{i\tau+h} =  y_{i\tau+h}
       -\widehat{y}_h(\widehat \beta,y_i,x_i)
       - \left( \frac{1}{n}\sum_{j=1}^n\ \frac{\partial \widehat{y}_h(\widehat{\beta},y_j,x_j)} {\partial \beta'} \right)
       \widehat \psi(\widehat \beta, y_i,x_i)  ,
$$
and a consistent estimator for  $\bar{\sigma}^{*2}_n$ is given by
$$
 \frac{1}{n} \sum_{i=1}^n \left(\widehat{u}^*_{i\tau+h} - \overline{\widehat{u}}^*_{\tau+h}\right)^2,
$$
where  $\overline{\widehat{u}}^*_{\tau+h} =     \frac 1 n \sum_{i=1}^n  \widehat{u}^*_{i\tau+h} $. If $\widehat \beta$ is not a method of moment
estimator (e.g., it is a more general GMM estimator), then the formula for  $\widehat \psi(\widehat \beta, y_i,x_i)  $ needs to 
be generalized accordingly, but the final estimation of $\bar{\sigma}^{*2}_n$ by the sample variance of $ \widehat u^*_{i\tau+h} $ remains
unchanged.

\newpage

\section{Online Appendix}

The Online Appendix contains extensions, additional simulations and empirical replications. We first discuss how to extend the procedure when there exists an untreated group (in Online Appendix \ref{controlgroup}). The simulations in the main text consider the case of homogeneous time trends and autoregressive parameters. We relax this in Online Appendix \ref{sec:Heterogeneouscoeffs}. Finally, we present two additional empirical applications in Online Appendix \ref{sec:overdose} and \ref{sec:DIDrefugees}.

\subsection{Extension: untreated group} \label{controlgroup}

In this section, we discuss how to modify our baseline procedure when a group of individuals not exposed to the treatment is available.

Without an untreated group, Section \ref{sec:Baseline-model} derived sufficient conditions ensuring that ${\rm FAT}_{h}$ defined in \eqref{DefFAT} is a consistent and asymptotically normal estimator of ${\rm ATT}_{h}$ defined in \eqref{ATT}. These conditions are the ability to obtain forecasts of the counterfactuals using pre-treatment data that are on average unbiased (Assumption \ref{Unbiasedness}) and the validity of a central limit theorem (Assumption \ref{ass:Sampling}). As discussed in the main text, these conditions exclude the presence of time effects such as macro shocks that affect all individuals between times $\tau$ and $\tau+h$, $h\geq 1$, and that are unforecastable using pre-treatment data. The presence of an untreated
group allows us to weaken this assumption. 

\subsubsection{DFAT: FAT with an untreated group}

Suppose that all individuals are untreated before the implementation
of the treatment at time $\tau$ and that some individuals remain untreated after $\tau$. Let $D_{i}=1$ if individual $i$ is treated after $\tau$. The observed outcome of individual $i$ at
time $t$ is then

\begin{equation}
y_{it}=D_{i}\left[1\left(t\leq\tau\right)y_{it}\left(0\right)+1\left(t>\tau\right)y_{it}\left(1\right)\right]+\left(1-D_{i}\right)y_{it}\left(0\right).\label{obsoutcome_control}
\end{equation}
As before, the parameter of interest is the average treatment effect on the
treated $h$ periods after the implementation of the treatment: 
\begin{equation}
{\rm ATT}_{h}=\frac{1}{n}\sum^n_{i=1}\mathbb{E}\left(\left.y_{i\tau+h}\left(1\right)-y_{i\tau+h}\left(0\right)\right\vert D_{i}=1\right).\label{ATT_control}
\end{equation}
Our proposed estimator is defined as:
\begin{equation}
\widehat{\rm DFAT}_{h}= \frac{1}{n_1}\sum_{i:D_i=1}(\left.y_{i\tau+h}-\widehat{y}_{i\tau+h}(0)\right) -\frac{1}{n_0}\sum_{i:D_i=0}(\left.y_{i\tau+h}-\widehat{y}_{i\tau+h}(0)\right),
\label{DFAThat}
\end{equation}
where $n_1$ is the number of treated individuals at time $\tau+h$, $n_0$ is the number of untreated individuals at time $\tau+h$, and  $y_{i\tau+h}$ is the observed outcome at $\tau+h$ given by \eqref{obsoutcome_control}.

Note that under \eqref{averagebias} below, $\mathbb{E}(\widehat{\rm DFAT}_h)={\rm ATT}_h$ in \eqref{ATT_control}:
\begin{equation}
\frac{1}{n}\sum^n_{i=1}\mathbb{E}\left(\left.y_{i\tau+h}\left(0\right)-\widehat{y}_{i\tau+h}(0)\right\vert D_{i}=1\right)=\frac{1}{n}\sum^n_{i=1}\mathbb{E}\left(\left.y_{i\tau+h}\left(0\right)-\widehat{y}_{i\tau+h}(0)\right\vert D_{i}=0\right).
\label{averagebias}
\end{equation}

Unlike in the baseline case, the forecast $\widehat{y}_{i\tau+h}(0)$ can 
be biased, as long as the average bias for the treated group equals the average bias for
the untreated group. As a consequence, the DGP for $y_{it}(0)$ can contain
additive time effects that are common across individuals such as additive macro
shocks that affect both treated and untreated groups in the same way.

Condition \eqref{averagebias} is satisfied when the DGP components in Assumption~\ref{ass:Series2} have the same average conditional expectation across treatment groups, i.e., when the mean-stationary, unit-root, and deterministic trend components are balanced on average between treated and untreated units. This is distinct from parallel trends: parallel trends restricts the evolution of outcome differences over time, while our condition restricts balance in the components that determine forecast errors.

The presence of an untreated group allows us to substitute Assumption \ref{ass:Series2} to allow for a common shock that is not necessarily polynomial, e.g., $y_{it}\left(0\right)=\tilde{y}_{it}(0) +\gamma_t$, where $\tilde{y}_{it}(0)$ satisfies Assumption \ref{ass:Series2}. Then, under assumptions similar to Assumptions \ref{Unbiasedness} and \ref{ass:Sampling}, it is possible to show consistency and asymptotic normality of \eqref{DFAThat}.

As in Section \ref{sec:Baseline-model} in the main text, we suggest using $\widehat{y}_{i\tau+h}^{(q_{i},R_{i})}(0)$ as an estimator for $\widehat{y}_{i\tau+h}(0)$ in \eqref{DFAThat}. Here, the parameters $q_{i},R_{i}$ do not necessarily have to be the same for the treated and untreated units.

\subsubsection{Comparison with Difference-in-Differences}

Despite the apparent similarity with the difference-in-differences (DiD) estimator, our method in the presence of an untreated group allows for DGPs for counterfactuals with more general forms of latent heterogeneity and outcome dynamics. For example, our approach allows for counterfactuals that follow fully heterogeneous autoregressive processes and/or unit root processes. In addition, it allows for the DGPs to have additive individual-specific time trends, as long as the deterministic time trend is either known or can be approximated by, e.g., a polynomial.

To see this, consider for example the following DGP for the counterfactuals: 
\[
y_{it}\left(0\right)=\rho y_{it-1}\left(0\right)+\gamma_{t}+k_{i}t+\epsilon_{it},\quad \mathbb{E}\left(\epsilon_{it}\right)=0,
\]
where $\rho\in[0,1]$, $\gamma_{t}$ is a common shock, and $k_{i}$ is an individual-specific time trend coefficient.

DiD can accommodate such specifications as long as the assumption of parallel-paths holds, which requires restricting the heterogeneity of both the initial condition, i.e., 
\[
\frac{1}{n}\sum^n_{i=1}\mathbb E(\left.y_{i0}(0) \right|D_i=1)=\frac{1}{n}\sum^n_{i=1}\mathbb E(\left.y_{i0}(0) \right|D_i=0),
\]
and the time trend coefficients, i.e.,
\[
\frac{1}{n}\sum^n_{i=1}k_{i}I\left(D_{i}=1\right)=\frac{1}{n}\sum^n_{i=1}k_{i}I\left(D_{i}=0\right),
\]
where $I(\cdot)$ is the indicator function. 

In contrast, ${\rm DFAT}_h$ does not require restricting the unobserved individual heterogeneity, and allows for heterogeneous $k_i$. In addition, it is straightforward to include lagged pre-treatment covariates with a homogeneous autoregressive parameter or a heterogeneous one, which is generally considered problematic in DiD methods. 

\medskip
\noindent \textbf{Remark (Dynamics and Transformed Outcomes).} 
It is theoretically possible to reconcile outcome dynamics with parallel trends by applying the parallel trends assumption to a transformed outcome. As emphasized by recent literature (and discussed in \cite{MarxTamerTang} in the context of sequential exchangeability), if the untreated outcome follows a dynamic process
\[
Y_{it}(0) = \alpha_i + \rho Y_{i,t-1}(0) + \lambda_t + u_{it}, \quad \mathbb{E}[u_{it}] = 0,
\]
one can define a transformed variable $w_{it}(\rho) \equiv Y_{it} - \rho Y_{i,t-1}$. If parallel trends hold for $w_{it}(\rho)$---i.e., $\mathbb{E}[w_{it}(\rho) - w_{is}(\rho) \mid D_i] = \mathbb{E}[w_{it}(\rho) - w_{is}(\rho)]$---then a standard DiD estimator applied to this transformed outcome identifies the ATT. Given enough pre-treatment periods, $\rho$ can be estimated consistently via Anderson-Hsiao.

However, the class of data generating processes handled by the FAT framework is broader in an important sense. The dynamic DiD specification above, as well as standard GMM dynamic-panel estimators \citep{ArellanoBond1991, BlundellBond1998}, typically require a \emph{homogeneous} autoregressive coefficient $\rho$. Even in flexible short-panel frameworks like \cite{ArellanoBonhomme2012}, the coefficient on the lagged outcome is often restricted. In contrast, FAT accommodates models of the form
\[
Y_{it}(0) = \rho_i Y_{i,t-1}(0) + \text{trend}_i + U_{it},
\]
where the autoregressive persistence $\rho_i$ varies across units. These types of models fall outside the scope of most existing short-panel estimators, which typically require restrictions on the form of dynamic heterogeneity.
\medskip

In addition, our paradigm of first obtaining individual-specific forecasts and only afterwards averaging across individuals avoids any concern about estimating weighted as opposed to unweighted treatment effects. It is well known that for unbalanced panels and for staggered adoption designs the DiD method will estimate a weighted average of the individual specific treatment effects with weights that are determined implicitly by the regression design, and that may even be negative in some cases (see e.g.\ \citealt{deChaisemartinX2020,GoodmanBacon}). By contrast, FAT (or DFAT) explicitly gives a weight of $1/n$ (or $1/n_d$) to each individual by construction, in accordance with the unweighted treatment effect that is specified as estimand.
  
\subsubsection{Related literature}

Conceptually, when there exists an untreated group, our solution resembles difference-in-differences or, more generally, an ``event-study design analysis'', e.g., \cite{Borusyaketal2021}. Although the estimated outcome equations may look similar, there is an important distinction between these methods and ours. For example, the extension of FAT to the case of a control group (which we call DFAT) uses control groups to correct for the effect of a common shock, whereas the other methods use control groups to correct for selection into treatment (under different assumptions). Additionally, FAT allows for heterogeneous nonlinear time trends as well as for heterogeneous effects of lagged pre-treatment outcomes. In contrast, there is no straightforward way to control for pre-treatment lagged outcomes in the specifications of, e.g., \cite{SunAbraham, CallawaySantAnna, Borusyaketal2021}. In fact, recent literature shows that in the presence of outcome dynamics, conventional estimators used in this literature produce biased estimates for the treatment effect \citep{BotosaruLiu2025,Cornwall2025}.

The synthetic control method has been increasingly used to evaluate the effect of
interventions implemented at an aggregate level, see \cite{Abadie2021} for a recent review. In the conventional
setting for synthetic controls, there is only one unit that is treated and there are many untreated units that could be used as pseudo-controls, i.e. these are untreated units selected such that the weighted average of their past outcomes ``resembles'' the trajectory of past outcomes of the treated unit. The counterfactual outcome for the treated unit is then constructed as a weighted average of the post-treatment outcomes of the selected pseudo-control units. In comparison, in our baseline setting, all individuals in the population are treated and there are no control units. The counterfactual outcome for each treated unit is a weighted average of the unit's \textit{own} past outcomes. The properties of our estimator rest on
averaging across many treated units, an advantage of which is standard inference. Our results 
apply even when the number of pre-treatment time periods is small, and we fully characterize the class
of DGPs that obtains a consistent estimator of the ATT.

Imputing counterfactual outcomes for the treated from data on the control is used in the literature on matrix completion, e.g., \cite{athey2021matrix, bai2019matrix, fernandez2020low}. Our framework has a thin matrix of outcomes. However since we do not observe cross-sectional control units, we cannot appeal to these methods to impute the counterfactuals.

\subsection{Simulation: Heterogeneous coefficients}\label{sec:Heterogeneouscoeffs}

We compare the finite-sample behavior of the Unobserved-Components FAT when the counterfactual process satisfies Assumption \ref{ass:Series2} with heterogeneous coefficients $\rho_i$ and $\delta_i$. That is, we consider the same specification of $y_{it}(0)$ as in Section \ref{indicator} with $I_1=I_2=I_3=1$, with the only changes being that $\rho_i$ and $\delta_i$ vary across individuals.

Table \ref{tab:table3} presents the results for  $\delta_i \sim \mathcal{U}\left[0,2\right]$ with a homogeneous autoregressive parameter $\rho=0.2$ (top panel) and with a heterogeneous autoregressive parameter $\rho_i \sim \mathcal{U}\left[0,0.99\right]$ (bottom panel). We can see that the presence of heterogeneous parameters does not change the conclusions that we derived from Table \ref{tab:table2}.

\begin{table}[]
\centering
\scriptsize
\begin{tabular}{@{}crlllll@{}}
\toprule
\multicolumn{1}{l}{}                                                                 & \multicolumn{1}{l}{}      & \multicolumn{5}{c}{$R_i$}             \\ \cmidrule(l){3-7} 
                                                                                     & \multicolumn{1}{l}{}      & $q+1$ & $q+2$ & $q+3$ & $q+4$ & $q+5$ \\ \midrule
 \multirow{9}{*}{\begin{tabular}[c]{@{}c@{}}Stationary AR(1) \\ $I_1=I_2=I_3=1$ \\ with $\delta_i \sim \mathcal{U}[0,2]$ \end{tabular}} & \multicolumn{1}{l}{$q=0$} &       &       &       &       &       \\ \cmidrule(lr){2-2}
 & bias                      & 0.999 & 1.4987  & 2.0003  & 2.5012 & 3.0015  \\
 & s.e.                      & 0.0544 & 0.0585 & 0.0655  & 0.0739 & 0.0826  \\ \cmidrule(lr){2-2}
 & \multicolumn{1}{l}{$q=1$} &        &         &         &         &         \\ \cmidrule(lr){2-2}
 & bias                      & 0.0013 & 0.0038 &  0.0009 & 0 &         \\
 & s.e.                      & 0.0793 & 0.0666  & 0.0646 & 0.0635  &         \\ \cmidrule(lr){2-2}
 & \multicolumn{1}{l}{$q=2$} &        &         &         &         &         \\ \cmidrule(lr){2-2}
 & bias                      & -0.0024 & 0.0065 & 0.0045  &         &         \\
 & s.e.                      & 0.1437 & 0.0971  & 0.0846  &         &       \\ \midrule
 \multirow{9}{*}{\begin{tabular}[c]{@{}c@{}} Stationary AR(1) \\ $I_1=I_2=I_3=1$ \\ with $\delta_i \sim \mathcal{U}[0,2]$ \\ and $\rho_i \sim \mathcal{U}[0,0.99]$  \end{tabular}} & \multicolumn{1}{l}{$q=0$} &       &       &       &       &       \\ \cmidrule(lr){2-2}
 & bias                      & 0.9967  & 1.4961  & 1.9963  & 2.4958  & 2.9947  \\
 & s.e.                      & 0.0554 & 0.0613  & 0.0693  & 0.078  & 0.0861  \\ \cmidrule(lr){2-2}
 & \multicolumn{1}{l}{$q=1$} &        &         &         &         &         \\ \cmidrule(lr){2-2}
 & bias                      & -0.0002 & 0.001 & 0.0027 & 0.0035 &         \\
 & s.e.                      & 0.0717 & 0.0629 & 0.0643 & 0.0672 &         \\ \cmidrule(lr){2-2}
 & \multicolumn{1}{l}{$q=2$} &        &         &         &         &         \\ \cmidrule(lr){2-2}
 & bias                      & 0.0019  & -0.0023 & -0.0005  &         &         \\
 & s.e.                      & 0.1264  & 0.0929  & 0.0801  &         &         \\ \bottomrule \bottomrule
\end{tabular}
\caption{Bias and standard error (s.e.) for the Unobserved-Components FAT when the counterfactual is specified as in the left-most column. The time trend component is heterogeneous across individuals (top panel), with the addition of a cross-sectionally heterogeneous autoregressive component (bottom panel). Stationary initial condition for the AR($1$) component for each $i$.}
\label{tab:table3}
\end{table}

\subsection{Empirical replication: Overdose mortality and medical cannabis laws}\label{sec:overdose}

We use data from \cite{Shover} to analyze the effect of adopting
legalized medical cannabis laws on opioid overdose mortality in the US. \cite{Shover} contributes to the debate about whether the adoption of such laws has decreased overdose mortality (see, e.g., \citealt{Bachuber}).

The unit of observation is at the level of state-year, with states adopting legalized medical cannabis laws from 1999 to 2017. Our analysis includes $n=34$ states that legalized medical cannabis during that period.\footnote{In our sample, $40$ states become treated during this time interval. We drop Hawaii, Colorado, and Nevada because they have too few pre-treatment observations (one and two) and Indiana, North Dakota, and West Virginia because they do not have any post-treatment observations.} The outcome of interest is the log mortality rate.

Both \cite{Bachuber} and \cite{Shover} use two-way fixed-effects estimators to analyze the effects of adopting legalized medical cannabis laws. In a staggered adoption setting, this estimation method produces biased results, e.g. \cite{GoodmanBacon}. Therefore, we first redo the analysis to remove the bias of the original studies by using various methods, such as the staggered DiD approach of \cite{CallawaySantAnna} and the generalized synthetic control method of \cite{Xu}. 
 Figure \ref{fig:CallawaySantAnna} uses the method of \cite{CallawaySantAnna} with not-yet-treated-states as control units, while Figure \ref{fig:Xu} uses the method of \cite{Xu} to impute counterfactuals for each treated unit using a linear two-way fixed effects regression. This analysis finds an initial increase in overdose mortality and then a reversal, but neither effect is statistically significant. 

\begin{figure}
\begin{centering}
\includegraphics[width=0.5\linewidth]{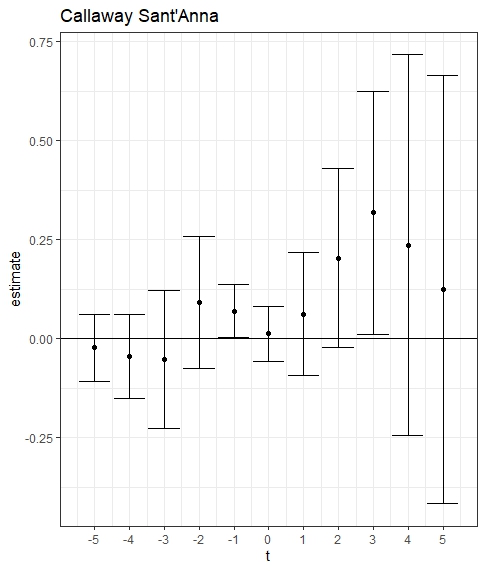}
\par\end{centering}
\caption{Overdose mortality rate as a function of time to adoption using not-yet-treated-states as control units.}
\label{fig:CallawaySantAnna}
\end{figure}

\begin{figure}
\begin{centering}
\includegraphics[width=0.5\linewidth]{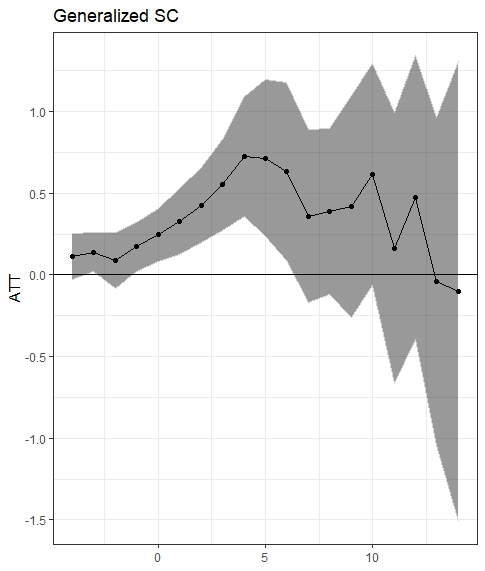}
\par\end{centering}
\caption{Overdose mortality rate as a function of time to adoption via generalized SC.}
\label{fig:Xu}
\end{figure}

We then implement FAT. We start with a plot of the log mortality rate averaged across states as a function of time to adoption; see Figure \ref{fig:timetoadoption}, in order to get a sense for the time series properties of the outcome of interest. The outcome appears to be non-stationary over the pre-treatment period, so we choose the smallest possible estimation window $R=q+1$ to calculate FAT (as in \eqref{eq:UC_FAT}) and let the order of the polynomial be $q=1,2$. Figure \ref{fig:Shover_FAT_placebo} shows our FAT estimates for the ATT as a function of the forecast horizon. Our estimates are relatively stable across the polynomial orders, and our results show a slight increase in the overdose mortality rate. However, the increase does not appear to be statistically significant across the forecast horizon. Our results thus corroborate those of the approaches of \cite{CallawaySantAnna} and \cite{Xu}, but without using a control group.

We also compute placebo FAT by assuming that the law was adopted by each state either one year earlier or two years earlier,  respectively, in Figure \ref{fig:Shover_FAT_placebo}. We consider these values because FAT is computed using the outcome either one period before adoption (when $q=1$) or two periods before adoption (when $q=2$). Figure \ref{fig:Shover_FAT_placebo} shows the placebo FAT estimates computed with $R=q+1$, $q=1,2$. The forecast horizon is $h=1$. The figure shows that the estimated placebo FATs are statistically insignificant. Although this is not a test of our assumptions, these placebo estimates offer suggestive evidence that FAT estimates may be interpreted as the ATT across different forecast horizons.

\begin{figure}
\begin{centering}
\includegraphics[width=0.5\linewidth]{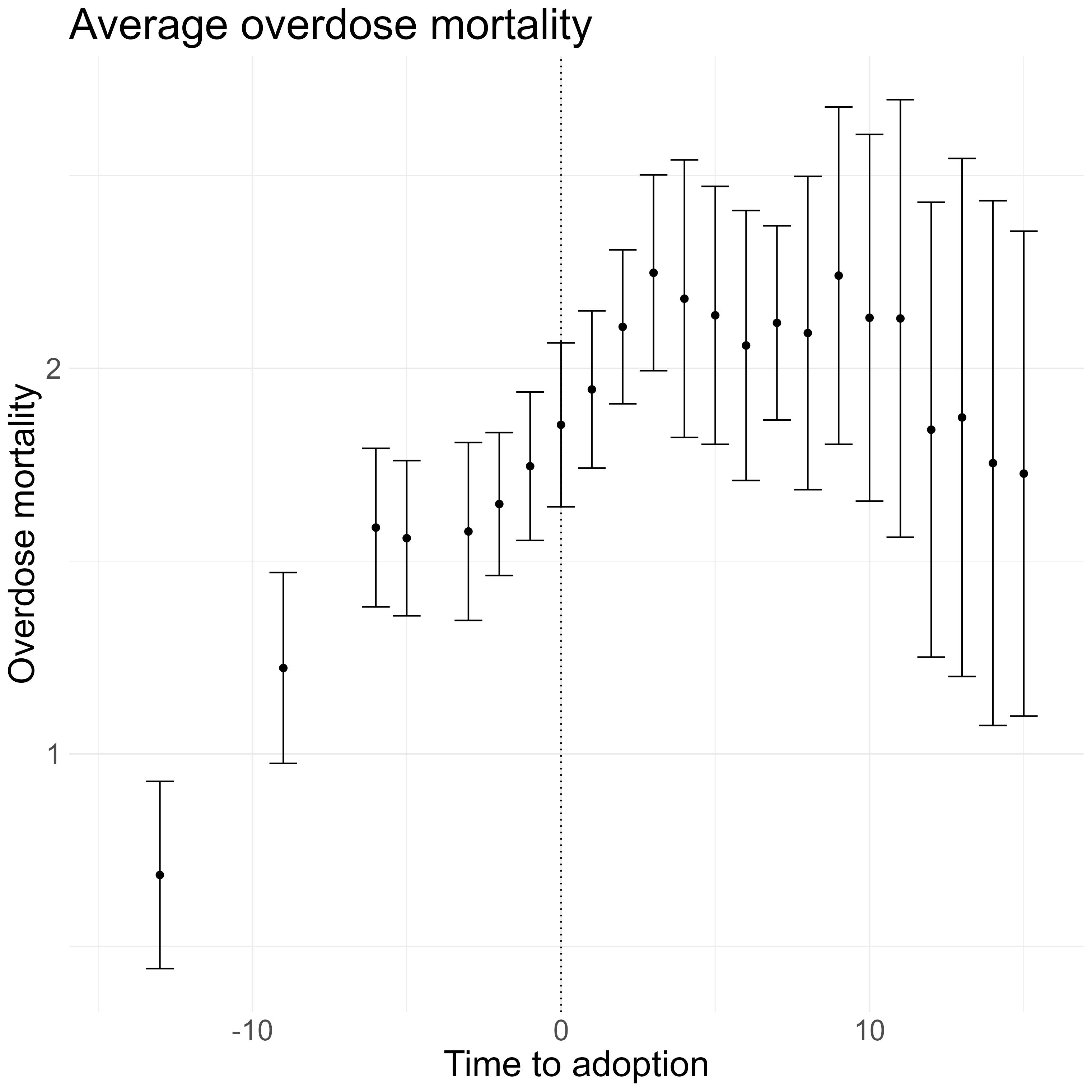}
\par\end{centering}
\caption{Log mortality rate averaged across states as a function
of time-to-adoption.}
\label{fig:timetoadoption}
\end{figure}

\begin{figure}[tb]
\begin{centering}
\includegraphics[width=0.6\linewidth]{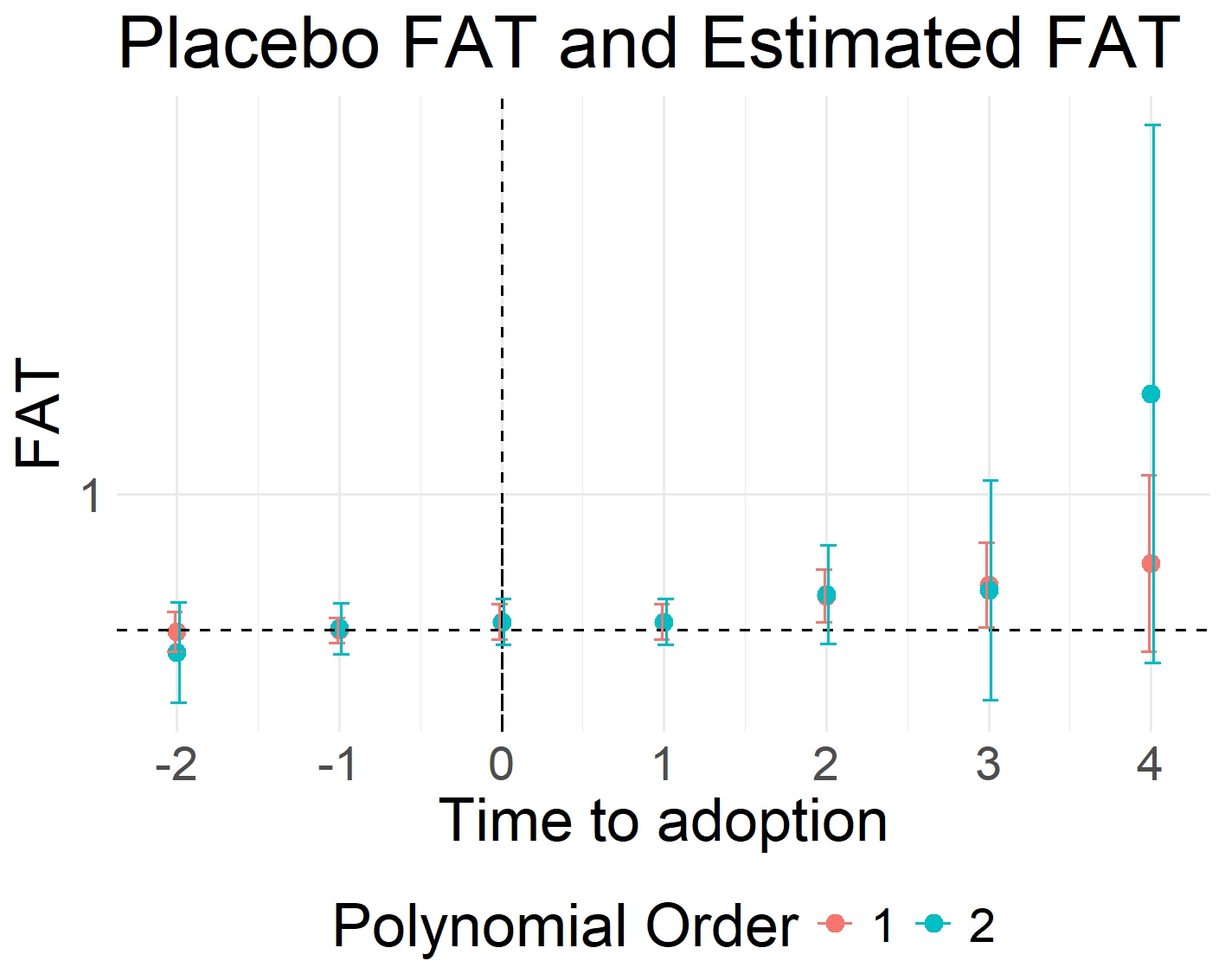}
\par\end{centering}
\caption{The effect of adopting legalized marijuana laws on overdose mortality rate. Estimates of FAT (to the right of the dashed vertical line) are shown as a function of forecast horizon (horizontal axis) and polynomial order (colored lines). FAT is computed with a variable number of pre-treatment time periods, $R=q+1, q=1,2$. Estimates of placebo FAT (to the left of the dashed vertical line) are shown as a function of lag (horizontal axis) and polynomial order (colored lines). Lag 0 is the actual adoption year. Placebo FAT is computed with a variable number of pre-treatment time periods, $R=q+1, q=1,2$.}
\label{fig:Shover_FAT_placebo}
\end{figure}

\subsection{Empirical replication: Refugees and far-right support}\label{sec:DIDrefugees}

In this replication exercise, we use data from \cite{Dinas}
which examines the relationship between refugee arrivals and support
for the far right. \cite{Dinas} consider the case of Greece,
and make use of the fact that some Greek islands (those close to the
Turkish border) witnessed sudden and unexpected increases in the number
of refugees during the summer of 2015, while other nearby Greek islands
saw much more moderate inflows of refugees. The municipalities in
the former Greek islands are considered treated, while the municipalities
in the latter are considered controls. The authors use a standard DiD
analysis to assess whether the treated municipalities were more supportive
of the far-right Golden Dawn party in the September 2015 general election.
The original data set contains a total of 96 municipalities, 48 of which were treated, and data on four elections: three pre-treatment elections in 2012, 2013, 2015, and one post-treatment election in 2016.
The outcome of interest is the vote share for Golden Dawn (GD). Figure \ref{fig:Dinas_PP} shows the vote share for GD averaged across municipalities, treated and control, before and after the treatment time.

\begin{figure}
\begin{centering}
\includegraphics[width=0.6\linewidth]{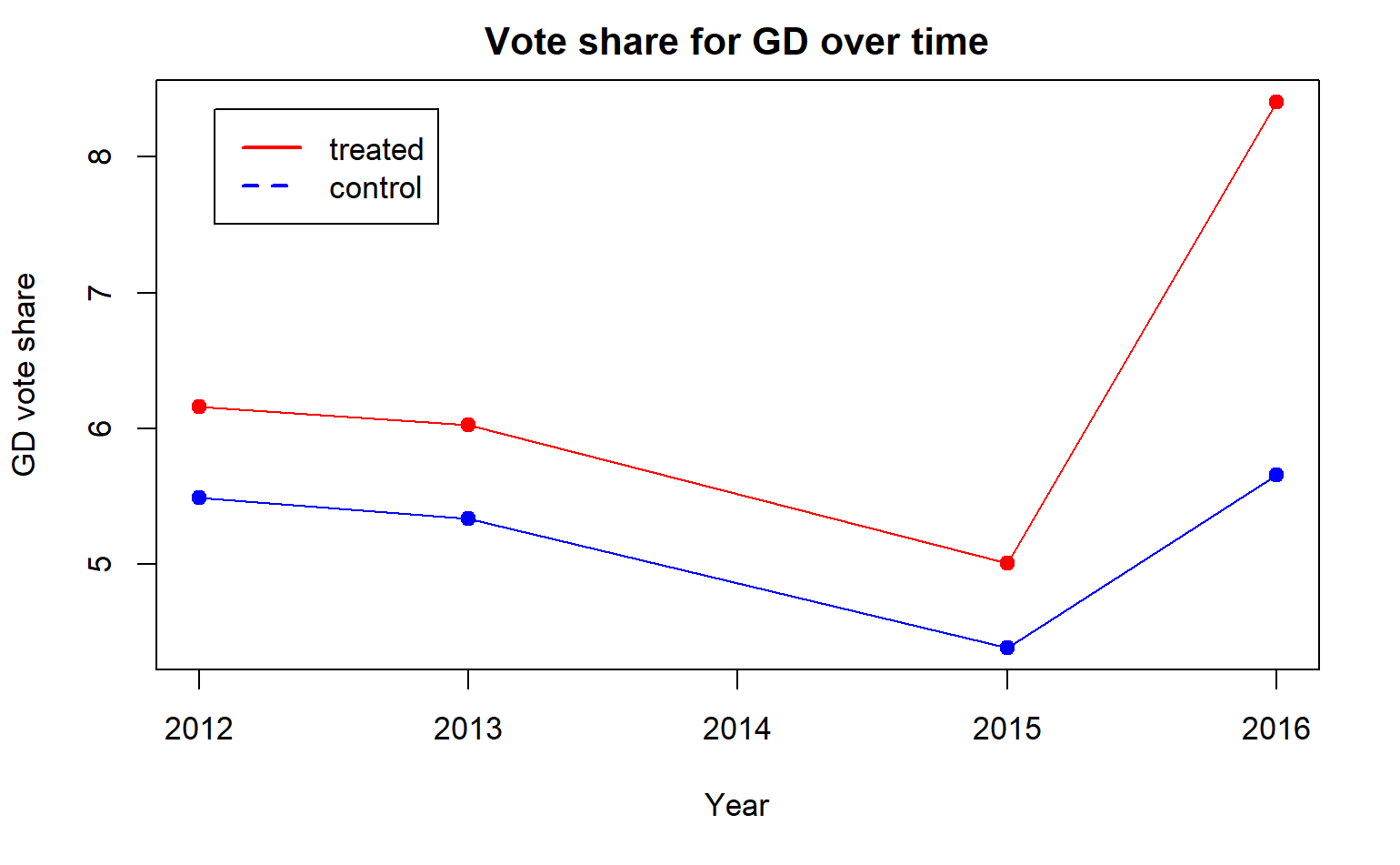}
\par\end{centering}
\caption{Vote share for Golden Dawn averaged across municipalities before and after $2015$ for municipalities that were treated (red) and control (blue).}
\label{fig:Dinas_PP}
\end{figure}

We use data on both the treated and the control municipalities to
compute $\widehat{\rm{DFAT}}_h$ in equation (\ref{DFAThat}) with $h=1$ and show that our estimate replicates
the original DiD estimates. 
This application can be viewed as a ``worst-case''
scenario for our proposed estimator since the number of treated units
is very small. We show results that use all three pre-treatment elections,
in which case the order of the polynomial is $q_{i}=q\in\left\{0,1,2\right\} $,
and results that use only the $2013$ and $2015$ pre-treatment elections,
in which case the order of the polynomial is $q_{i}=q\in\left\{ 0,1\right\} $.
Note that we perform municipality-specific polynomial regressions
to compute the forecasted vote share -- the counterfactual outcome
of interest, using the same polynomial order across all municipalities. 

As Table \ref{tab:Dinas} shows, our DFAT results are comparable with those in
the original paper. The DiD estimates in the original paper are $0.0206$
and $0.0208$ when using $2013$ and $2015$ as pre-treatment periods
and all pre-treatment periods, respectively. The two-way fixed-effects
estimate is $0.021$ with a standard error of $0.0393$.

We include here placebo FAT estimates which we compute by assuming that the treatment took place in $2015$. The point here is to assess if the estimated placebo FAT is statistically indistinguishable from zero. We use a polynomial of degree $q=0$ when the pre-treatment year is $2013$ and a polynomial of degree $q=1$ when the pre-treatment years are $2012$ and $2013$.  The placebo FATs are all statistically insignificant: the placebo FAT for the treated when $q=0$ is $-0.010$ with a standard error of $0.013$, while when $q=1$ it is $-0.009$ with a standard error of $0.029$. The placebo FAT for the control when $q=0$ is $-0.010$ with a standard error of $0.012$, while when $q=1$ it is $-0.008$ with a standard error of $0.017$. The corresponding placebo DFATs are then: $0$ for $q=0$ with a standard error of $0.016$ and $-0.001$ for $q=1$ with a standard error of $0.034$.

\begin{table}[]
\scriptsize
\centering
\begin{tabular}{@{}lllll@{}}
\cmidrule(r){1-4}
                 & FAT Treated                                              & FAT Control                                              & DFAT  &  \\ \cmidrule(r){1-4}
Polynomial order & \multicolumn{3}{c}{2013-2015}                                                                                               &  \\ \cmidrule(r){1-4}
$q=0$              & \begin{tabular}[c]{@{}l@{}}0.029 \\ (0.016)\end{tabular} & \begin{tabular}[c]{@{}l@{}}0.008 \\ (0.009)\end{tabular} & 0.021 &  \\
$q=1$              & \begin{tabular}[c]{@{}l@{}}0.054 \\ (0.028)\end{tabular} & \begin{tabular}[c]{@{}l@{}}0.032 \\ (0.026)\end{tabular}  & 0.022 &  \\ \cmidrule(r){1-4}
                 & \multicolumn{3}{c}{2012-2015}                                                                                               &  \\ \cmidrule(r){1-4}
$q=0$              & \begin{tabular}[c]{@{}l@{}}0.027 \\ (0.012)\end{tabular} & \begin{tabular}[c]{@{}l@{}}0.006 \\ (0.011)\end{tabular} & 0.021 &  \\
$q=1$              & \begin{tabular}[c]{@{}l@{}}0.038 \\ (0.025)\end{tabular} & \begin{tabular}[c]{@{}l@{}}0.017 \\ (0.016)\end{tabular} & 0.019  &  \\
$q=2$              & \begin{tabular}[c]{@{}l@{}}0.053 \\ (0.036)\end{tabular} & \begin{tabular}[c]{@{}l@{}}0.030 \\ (0.027)\end{tabular} & 0.023 &  \\ \cmidrule(r){1-4}
\end{tabular}
\caption{DFAT under different polynomial orders and pre-treatment periods.}
\label{tab:Dinas}
\end{table}

\subsection{Additional simulation results}

We consider the same DGP as that corresponding to Table \ref{tab:table2} in Section \ref{indicator}. We report simulation results for a smaller sample size, $N=30$, in Table \ref{tab:table_smallN}. This particular sample size is approximately that of two of our applications (in two of our applications, the sample size is $34$). The simulation results show that the large standard errors seen in the applications may be due to the small sample size. We note here that the size of the standard error of our estimator seen is comparable to that obtained in the studies that we replicate.

\begin{table}[]
\spacingset{1.7}
\centering
\tiny
\begin{tabular}{@{}crlllll@{}}
\toprule
\multicolumn{1}{l}{}                                                                 & \multicolumn{1}{l}{}      & \multicolumn{5}{c}{$R$}             \\ \cmidrule(l){3-7} 
                                                                                     & \multicolumn{1}{l}{}      & $q+1$ & $q+2$ & $q+3$ & $q+4$ & $q+5$ \\ \midrule
\multirow{9}{*}{\begin{tabular}[c]{@{}c@{}}Stationary AR(1) \\ $I_1=1, I_2 = 0 = I_3$ \end{tabular}}     & \multicolumn{1}{c}{$q=0$} &       &       &       &       &       \\ \cmidrule(lr){2-2}
 & bias                      & -0.0001 & -0.0023 & 0.0025  & 0.0011 & 0.0019 \\
 & s.e.                      & 0.2396 & 0.2139   & 0.2056  & 0.2056  & 0.2044  \\ \cmidrule(lr){2-2}
 & \multicolumn{1}{l}{$q=1$} &        &         &         &         &         \\ \cmidrule(lr){2-2}
 & bias                      & 0.0042 & -0.0097  & -0.0008  & -0.0016  &         \\
 & s.e.                      & 0.4027 & 0.3105  & 0.2745  & 0.2574  &         \\ \cmidrule(lr){2-2}
 & \multicolumn{1}{c}{$q=2$} &        &         &         &         &         \\ \cmidrule(lr){2-2}
 & bias                      & 0.0252 & -0.0144  & -0.005  &         &         \\
 & s.e.                      & 0.7297 & 0.4874  & 0.3984  &         &         \\ \midrule
 \multirow{9}{*}{\begin{tabular}[c]{@{}c@{}}Stationary AR(1) \\ + unit root \\ $I_1=1=I_2, I_3=0$ \end{tabular}}     & \multicolumn{1}{c}{$q=0$} &       &       &       &       &       \\ \cmidrule(lr){2-2}
 & bias                      & -0.0093 & 0.0003  & 0.0031  & 0.00045  & -0.0015 \\
 & s.e.                      & 0.3102 & 0.3031  & 0.3118  & 0.3315  & 0.3467  \\ \cmidrule(lr){2-2}
 & \multicolumn{1}{l}{$q=1$} &        &         &         &         &         \\ \cmidrule(lr){2-2}
 & bias                      & -0.0285 & -0.0149  & -0.0005  & -0.0023  &         \\
 & s.e.                      & 0.4966 & 0.4141  & 0.3641  & 0.3533  &         \\ \cmidrule(lr){2-2}
 & \multicolumn{1}{c}{$q=2$} &        &         &         &         &         \\ \cmidrule(lr){2-2}
 & bias                      & -0.0489 & -0.045  & -0.0156  &         &         \\
 & s.e.                      & 0.8543 & 0.6345  & 0.5282  &         &         \\ \midrule
 \multirow{9}{*}{\begin{tabular}[c]{@{}c@{}}Stationary AR(1) \\+ linear trend  \\ $I_1=1=I_3, I_2=0$ \end{tabular}} & \multicolumn{1}{l}{$q=0$} &       &       &       &       &       \\ \cmidrule(lr){2-2}
 & bias                      & 0.9999 & 1.4977 & 2.0025  & 2.5011 & 3.0019  \\
 & s.e.                      & 0.2396 & 0.2139 & 0.2056  & 0.2056 & 0.2044  \\ \cmidrule(lr){2-2}
 & \multicolumn{1}{l}{$q=1$} &        &         &         &         &         \\ \cmidrule(lr){2-2}
 & bias                      & 0.0042  & -0.0097 & -0.0008 & -0.0016 &         \\
 & s.e.                      & 0.4027 & 0.3105  & 0.2745  & 0.2574  &         \\ \cmidrule(lr){2-2}
 & \multicolumn{1}{l}{$q=2$} &        &         &         &         &         \\ \cmidrule(lr){2-2}
 & bias                      & 0.0252 & -0.0144  & -0.005  &         &         \\
 & s.e.                      & 0.7297 & 0.4874  & 0.3984  &         &         \\ \midrule
 \multirow{9}{*}{\begin{tabular}[c]{@{}c@{}}Stationary AR(1) \\ + linear trend \\ + unit root \\ $I_1=I_2=I_3=1$  \end{tabular}} & \multicolumn{1}{l}{$q=0$} &       &       &       &       &       \\ \cmidrule(lr){2-2}
 & bias                      & 0.9907 & 1.5003  & 2.0031  & 2.5004  & 3.0015  \\
 & s.e.                      & 0.3102 & 0.3031  & 0.3118  & 0.3315  & 0.3467  \\ \cmidrule(lr){2-2}
 & \multicolumn{1}{l}{$q=1$} &        &         &         &         &         \\ \cmidrule(lr){2-2}
 & bias                      & -0.0285 & -0.0149 & -0.0005 & 0.0023 &         \\
 & s.e.                      & 0.4966 & 0.4141 & 0.3641 & 0.3533 &         \\ \cmidrule(lr){2-2}
 & \multicolumn{1}{l}{$q=2$} &        &         &         &         &         \\ \cmidrule(lr){2-2}
 & bias                      & -0.0489   & -0.045    & -0.0156  &         &         \\
 & s.e.                      & 0.8543  & 0.6345   & 0.5282  &         &         \\ \bottomrule \bottomrule
\end{tabular}
\caption{Bias and standard error (s.e.) for the Unobserved-Components FAT when the counterfactual is specified as indicated in the left-most column. Sample size $N=30$.}
\label{tab:table_smallN}
\end{table}

\end{document}